\documentclass{amsart}%
\usepackage{amsfonts}
\usepackage{amsmath}
\usepackage{amsxtra}
\usepackage{amssymb}
\usepackage{graphicx}%
\providecommand{\U}[1]{\protect\rule{.1in}{.1in}}
\providecommand{\U}[1]{\protect\rule{.1in}{.1in}}
\providecommand{\U}[1]{\protect\rule{.1in}{.1in}}
\theoremstyle{plain}
\newtheorem{theorem}{Theorem}[section]
\newtheorem{lemma}[theorem]{Lemma}

\newtheorem{remark}{Remark}

\begin{document}
\title{Unstable and Stable Galaxy Models}
\author{Yan Guo}
\address{Lefschetz Center for Dynamical Systems, Division of Applied Mathematics, Brown
University, Providence, RI 02912, USA }
\email{guoy@cfm.brown.edu }
\author{Zhiwu Lin}
\address{Mathematics Department\\
University of Missouri\\
Columbia, MO 65211 USA}
\email{lin@math.missouri.edu}

\begin{abstract}
To determine the stability and instability of a given steady galaxy
configuration is one of the fundamental problems in the Vlasov theory for
galaxy dynamics. In this article, we study the stability of isotropic
spherical symmetric galaxy models $f_{0}(E)$, for which the distribution
function $f_{0}$ depends on the particle energy $E$ only. In the first part of
the article, we derive the first sufficient criterion for linear instability
of $f_{0}(E):$ $f_{0}(E)$ is linearly unstable if the second-order operator
\[
A_{0}\equiv-\Delta+4\pi\int f_{0}^{\prime}(E)\{I-\mathcal{P}\}dv
\]
has a negative direction, where $\mathcal{P}$ is the projection onto the
function space $\{g(E,L)\},$ $L$ being the angular momentum [see the explicit
formulae\ (\ref{ergodic-radial}) and (\ref{ergodic-spherical})]. In the second
part of the article, we prove that for the important King model, the
corresponding $A_{0}$ is positive definite. Such a positivity leads to the
nonlinear stability of the King model under all spherically symmetric perturbations.

\end{abstract}
\maketitle

\section{Introduction}

A galaxy is an ensemble of billions of stars, which interact by the
gravitational field which they create collectively. For galaxies, the
collisional relaxation time is much longer than the age of the universe
(\cite{BT}). The collisions can therefore be ignored and the galactic dynamics
is well described by the Vlasov - Poisson system (collisionless Boltzmann
equation)
\begin{equation}
\partial_{t}f+v\cdot\nabla_{x}f-\nabla_{x}U\cdot\nabla_{v}f=0,\text{
\ \ \ \ \ \ \ \ \ \ \ \ \ \ \ }\Delta U=4\pi\int_{\mathbf{R}^{3}}f(t,x,v)dv,
\label{vp}%
\end{equation}
where $\left(  x,v\right)  \in\mathbb{R}^{3}\times\mathbb{R}^{3}$, $f(t,x,v)$
is the distribution function and $U_{f}\left(  t,x\right)  $ is its
gravitational potential. The Vlasov-Poisson system can also be used to
describe the dynamics of globular clusters over their period of orbital
revolutions (\cite{fp84}). One of the central questions in such galactic
problems, which has attracted considerable attention in the astrophysics
literature, of \cite{bertin}, \cite{BT}, \cite{fp84}, \cite{merritt} and the
references there, is to determine \textit{dynamical stability }of steady
galaxy models. Stability study can be used to test a proposed configuration as
a model for a real stellar system. On the other hand, instabilities of steady
galaxy models can be used to explain some of the striking irregularities of
galaxies, such as spiral arms as arising from the instability of an initially
featureless galaxy disk (\cite{bertin}), (\cite{palmer}).

In this article, we consider stability of spherical galaxies, which are the
simplest elliptical galaxy models. Though most elliptical galaxies are known
to be non-spherical, the study of instability and dynamical evolution of
spherical galaxies could be useful to understand more complicated and
practical galaxy models . By Jeans's Theorem, a steady spherical galaxy is of
the form%
\[
f_{0}(x,v)\equiv f_{0}(E,L^{2}),
\]
where the particle energy and total momentum are
\[
E=\frac{1}{2}|v|^{2}+U_{0}(x),\ L^{2}=\left\vert x\times v\right\vert ^{2},
\]
and $U_{0}(x)=U_{0}\left(  \left\vert x\right\vert \right)  $ satisfies the
self-consistent Poisson equation. The isotropic models take the form
\[
f_{0}(x,v)\equiv f_{0}(E).
\]
The cases when $f_{0}^{\prime}(E)<0$ has been widely studied and these models
are known to be linearly stable to both radial (\cite{dbf73}) and non-radial
perturbations (\cite{ant62}). The well-known Casimir-Energy functional (as a
Liapunov functional)
\begin{equation}
\mathcal{H}(f)\equiv\int\int Q(f)+\frac{1}{2}\int\int|v|^{2}f-\frac{1}{8\pi
}\int|\nabla_{x}U_{f}|^{2}, \label{h}%
\end{equation}
is constant along the time evolution. If $f_{0}^{\prime}(E)<0,$ we can choose
the Casimir function $Q_{0}$ such that
\[
Q_{0}^{\prime}(f_{0}(E))\equiv-E
\]
for all $E.$ By a Taylor expansion of $\mathcal{H}(f)-\mathcal{H}(f_{0})$, it
follows that formally the first variation at $f_{0}$ is zero, that is,
$\mathcal{H}^{(1)}(f_{0}(E))=0$ (on the support of $f_{0}(E)$)$,$ and the
second order variation of $\mathcal{H}$ at $f_{0}$ is
\begin{equation}
\mathcal{H}_{f_{0}}^{(2)}[g]\equiv\frac{1}{2}\int\int_{\left\{  f_{0}%
>0\right\}  }\frac{g^{2}}{-f_{0}^{\prime}(E)}dxdv-\frac{1}{8\pi}\int
|\nabla_{x}U_{g}|^{2}dx \label{h2}%
\end{equation}
where $Q^{\prime\prime}(f_{0})=\frac{1}{-f_{0}^{\prime}(E)},\ g=f-f_{0}$ and
$\Delta U_{g}=\int gdv$. In the 1960s, Antonov (\cite{an61}, \cite{ant62})
proved that
\begin{equation}
\mathcal{H}_{f_{0}}^{(2)}[Dh]=\int\int\frac{\left\vert Dh\right\vert ^{2}%
}{\left\vert f_{0}^{\prime}(E)\right\vert }dxdv-\frac{1}{4\pi}\int\left\vert
\nabla\psi_{h}\right\vert ^{2}dx \label{Antonov}%
\end{equation}
is positive definite for a large class of monotone models. Here
\[
D=v\cdot\nabla_{x}-\nabla_{x}U_{0}\cdot\nabla_{v},
\]
$\ h(x,v)$ is odd in $v$ and $-\Delta\psi=\int Dhdv$. He showed that such a
positivity is equivalent to the linear stability of $f_{0}(E)$. In
\cite{dbf73}, Doremus, Baumann and Feix proved the radial stability of any
monotone spherical models. Their proof was further clarified and simplified in
\cite{dbf76}, \cite{sygnet84}, \cite{KS85}, and more recently in \cite{aly96},
\cite{gr-king}. In particular, this implies that any monotone isotropic models
are at least linearly stable.

Unfortunately, despite its importance and a lot of research (e.g.,
\cite{henon73}, \cite{bgh86}, \cite{bar71}, \cite{goodman88}), to our
knowledge, no rigorous and explicit instability criterion of non-monotone
models has been derived. When $f_{0}^{\prime}(E)$ changes sign, functional
$\mathcal{H}_{f_{0}}^{(2)}$ is indefinite and it gives no stability
information, although it seems to suggest that these models are not energy
minimizers under symplectic perturbations. In this paper, we first obtain the
following instability criterion for general spherical galaxies. For any
function $g$ with compact support within the support of $f_{0}(E),$ we define
the $\left\vert f_{0}^{\prime}(E)\right\vert -$weighted $L^{2}\left(
\mathbf{R}^{3}\times\mathbf{R}^{3}\right)  $ space $L_{\left\vert
f_{0}^{\prime}\right\vert }^{2}$ with the norm $\left\Vert \cdot\right\Vert
_{\left\vert f_{0}^{\prime}\right\vert }$ as
\begin{equation}
||h||_{|f_{0}^{\prime}|}^{2}\equiv\int\int|f_{0}^{\prime}(E)|h^{2}dxdv.
\label{weight}%
\end{equation}

\begin{theorem}
\label{theorem-insta}Assume that $f_{0}(E)$ has a compact support in $x$ and
$v,$ and $f_{0}^{\prime}$ is bounded. For $\phi\in H^{1},$ define the
quadratic form
\begin{equation}
(A_{0}\phi,\phi)=\int|\nabla\phi|^{2}dx+4\pi\int\int f_{0}^{\prime}(E)\left(
\phi-\mathcal{P}\phi\right)  ^{2}dxdv, \label{formula-qudratic-A0}%
\end{equation}
where $\mathcal{P}$ is the projector of $L_{\left\vert f_{0}^{\prime
}\right\vert }^{2}$ to
\[
\ker D=\left\{  g\left(  E,L^{2}\right)  \right\}  ,
\]
and more explicitly $\mathcal{P}\phi$ is given by (\ref{a0-radial}) for radial
functions and (\ref{ergodic-spherical}) for general functions. If there exists
$\phi_{0}\in H^{1}$ such that%
\begin{equation}
(A_{0}\phi_{0},\phi_{0})<0, \label{negative-con}%
\end{equation}
then there exists $\lambda_{0}>0$ and $\phi\in H^{2},$ $f\left(  x,v\right)  $
given by (\ref{f-exp}), such that $e^{\lambda_{0}t}[f,\phi]$ is a growing mode
to the Vlasov-Poisson system (\ref{vp}) linearized around $\left[
f_{0}(E),U_{f_{0}}\right]  .$
\end{theorem}

A similar instability criterion can be obtained for symmetry preserving
perturbations of anisotropic spherical models $f_{0}\left(  E,L^{2}\right)  $,
see Remark \ref{remark anistropic}. We note that the term $\mathcal{P}\phi$ in
the instability criterion is highly non-local and this reflects the collective
nature of stellar instability. The proof of Theorem \ref{theorem-insta} is by
extending an approach developed in \cite{lin01} for 1D Vlasov-Poisson, which
has recently been generalized to Vlasov-Maxwell systems (\cite{lw-linear},
\cite{lw-new}). There are two elements in this approach. One is to formulate a
family of dispersion operators $A_{\lambda}$ for the potential, depending on a
positive parameter $\lambda$. The existence of a purely growing mode is
reduced to find a parameter $\lambda_{0}$ such that the $A_{\lambda_{0}}$ has
a kernel. The key observation is that these dispersion operators are
self-adjoint due to the reversibility of the particle trajectories. Then a
continuation argument is applied to find the parameter $\lambda_{0}$
corresponding to a growing mode, by comparing the spectra of $A_{\lambda}$ for
very small and large values of $\lambda$. There are two new complications in
the stellar case. First, the essential spectrum of $A_{\lambda}$ is
$[0,+\infty)$ and thus we need to make sure that the continuation does not end
in the essential spectrum$.$This is achieved by using some compactness
property due to the compact support of the stellar model. Secondly, it is more
tricky to find the limit of $A_{\lambda}$ when $\lambda$ tends to zero. For
that, we need an ergodic lemma (Lemma \ref{lemma-ergodic}) and use the
integrable nature of the particle dynamics in a central field to derive an
expression for the projection $\mathcal{P}\phi$ appeared in the limit.

In the second part of the article, we further study the nonlinear (dynamical)
stability of the normalized King model:%

\begin{equation}
f_{0}=[e^{E_{0}-E}-1]_{+} \label{king}%
\end{equation}
motivated by the study of the operator $A_{0}.$ The famous King model
describes isothermal galaxies and the core of most globular clusters
\cite{king}. Such a model provides a canonical form for many galaxy models
widely used in astronomy. Even though $f_{0}^{\prime}<0$ for the King model,
it is important to realize that, because of the Hamiltonian nature of the
Vlasov-Poisson system (\ref{vp}), linear stability fails to imply nonlinear
stability (even in the finite dimensional case). The Liapunov functional is
usually required to prove nonlinear stability. In the Casimir-energy
functional (\ref{h}), it is natural to expect that the positivity of such a
quadratic form $\mathcal{H}_{f_{0}}^{(2)}[g]$ should imply stability for
$f_{0}(E)$. However, there are at least two serious mathematical difficulties.
First of all, it is very challenging to use the positivity of $\mathcal{H}%
_{f_{0}}^{(2)}[g]$ to control higher order remainder in $\mathcal{H}%
(f)-\mathcal{H}(f_{0})$ to conclude stability \cite{wang}. For example, one of
the remainder terms is $f^{3}$ whose $L^{2}$ norm is difficult to be bounded
by a power of the stability norm. The non-smooth nature of $f_{0}(E)$ also
causes trouble here$.$ Second of all, even if one can succeed in controlling
the nonlinearity, the positivity of $H_{f_{0}}^{(2)}[g]$ is \textit{only}
valid for certain perturbation of the form $g=Dh$ \cite{KS85}. It is not clear
at all if \textit{any} arbitrary, general perturbation can be reduced to the
form $Dh$. To overcome these two difficulties, a direct variational approach
was initiated by Wolansky \cite{wol}, then further developed systematically by
Guo and Rein in \cite{G1}, \cite{G2}, \cite{GR2}, \cite{GR3}, \cite{GR4}.
Their method avoids entirely the delicate analysis of the second order
variation $\mathcal{H}_{f_{0}}^{(2)}$ in (\ref{h2}), which has led to first
rigorous nonlinear stability proof for a large class of $f_{0}(E).$ The high
point of such a program is the nonlinear stability proof for every polytrope
\cite{GR3} $f_{0}(E)=(E_{0}-E)_{+}^{k}$. Their basic idea is to construct
galaxy models by solving a variational problem of minimizing the energy under
some constraints of Casimir invariants. A concentration-compactness argument
is used to show the convergence of the minimizing sequence. All the models
constructed in this way are automatically stable. \ 

Unfortunately, despite its success, the King model can not be studied by such
a variational approach. The Casimir function for a \textit{normalized} King
model is
\begin{equation}
Q_{0}(f)=(1+f)\ln(1+f)-1-f, \label{q}%
\end{equation}
which has very slow growth for $f\rightarrow\infty.$ As a result, the direct
variational method fails. Recently, Guo and Rein \cite{gr-king} proved
nonlinear radial stability among a class of measure-preserving perturbations
\begin{equation}
\mathcal{S}_{f_{0}}\equiv\left[  f(t,r,v_{r},L)\geq0:\;\int Q(f,L)=\int
Q(f_{0},L),\text{ for }Q\in C_{c}^{\infty}\text{ and }Q(0,L)\equiv0.\right]  .
\label{s}%
\end{equation}
The basic idea is to observe that for perturbations in the class
$\mathcal{S}_{f_{0}}$, one can write $g=f-f_{0}$ as $Dh=\left\{  h,E\right\}
$. Therefore, $\mathcal{H}_{f_{0}}^{(2)}[g]$ $=\mathcal{H}_{f_{0}}^{(2)}[Dh]$,
for which the positivity was proved in \cite{KS85} for radial perturbations.
To avoid the difficulty of controlling the remainder term by $\mathcal{H}%
_{f_{0}}^{(2)}[g]$, an indirect contradiction argument was used in
\cite{gr-king}.

As our second main result of this article, we establish nonlinear stability of
King's model for general perturbations with spherical symmetry:

\begin{theorem}
\label{theorem-stability}The King's model $f_{0}=[e^{E_{0}-E}-1]_{+}$ is
nonlinearly stable under spherically symmetric perturbations in the following
sense: given any $\varepsilon>0$ there exists $\varepsilon_{1}>0$ such that
for any compact supported initial data $f(0)\in C_{c}^{1}$ with spherical
symmetry, if $d\left(  f\left(  0\right)  ,f_{0}\right)  <\varepsilon_{1}$
then
\[
\sup_{0\leq t<\infty}d\left(  f\left(  t\right)  ,f_{0}\right)  <\varepsilon,
\]
where the distance functional $d\left(  f,f_{0}\right)  $ is defined by
(\ref{defn-distance}).
\end{theorem}

For the proof, we extended the approach in \cite{lw-nonlinear} for the
$1\frac{1}{2}D$ Vlasov-Maxwell model. To prove nonlinear stability, we study
the Taylor expansion of $\mathcal{H}(f)-\mathcal{H}(f_{0})$. Two difficulties
as mentioned before are: to prove the positivity of the quadratic form and to
control the remainder. We use two ideas introduced in \cite{lw-nonlinear}. The
first idea is to use any finite number of Casimir functional $Q_{i}\left(
f,L^{2}\right)  $ as constraints. The difference from \cite{gr-king} is that
we do not impose $Q_{i}\left(  f,L^{2}\right)  =Q_{i}\left(  f_{0}%
,L^{2}\right)  $ in the perturbation class, but expand the invariance equation
$Q_{i}\left(  f\left(  t\right)  ,L^{2}\right)  -Q_{i}\left(  f_{0}%
,L^{2}\right)  =Q_{i}\left(  f\left(  0\right)  ,L^{2}\right)  -Q_{i}\left(
f_{0},L^{2}\right)  $ to the first order. In this way, we get a constraint for
$g=f-f_{0}$ in the form that the coefficient of its projection to
$\partial_{1}Q_{i}\left(  f_{0},L^{2}\right)  $ is small. Putting these
constraints together, we deduce that a finite dimensional projection of $g$ to
the space spanned by $\left\{  \partial_{1}Q_{i}\left(  f_{0},L^{2}\right)
\right\}  $ is small. To control the remainder term, we use a duality
argument. Noting that it is much easier to control the potential $\phi$, we
use a Legendre transformation to reduce the nonlinear term in $g$ to a new one
in $\phi$ only. The key observation is that the constraints on $g$ in the
projection form are nicely suited to the Legendre transformation and yields a
non-local nonlinear term in $\phi$ only with the projections kept. By
performing a Taylor expansion of this non-local nonlinear term in $\phi$, the
quadratic form becomes a truncated version of $(A_{0}\phi,\phi)$ defined by
(\ref{formula-qudratic-A0}), whose positivity can be shown to be equivalent to
that of Antonov functional. The the remainder term now is only in terms of
$\phi$ and can be easily controlled by the quadratic form. The new
complication in the stellar case is that the steady distribution $f_{0}\left(
E\right)  $ is non-smooth and compactly supported. Therefore, we split the
perturbation $g$ into inner and outer parts, according to the support of
$f_{0}$. For the inner part, we use the above constrainted duality argument
and the outer part is estimated separately.

\section{An Instability Criterion}

We consider a steady distribution
\[
f_{0}\left(  x,v\right)  =f_{0}(E)
\]
has a bounded support in $x$ and $v$ and $f_{0}^{\prime}$ is bounded, where
the particle energy $E=\frac{1}{2}|v|^{2}+U_{0}(x).$ The steady gravitational
potential $U_{0}(x)$ satisfies a nonlinear Poisson equation%
\[
\Delta U_{0}=4\pi\int f_{0}dv.
\]

The linearized Vlasov-Poisson system is
\begin{equation}
\partial_{t}f+v\cdot\nabla_{x}f-\nabla_{x}U_{0}\cdot\nabla_{v}f=\nabla_{x}%
\phi\cdot\nabla_{v}f_{0},\text{ \ \ \ \ \ }\Delta\phi=4\pi\int f(t,x,v)dv.
\label{lvp}%
\end{equation}
A growing mode solution $(e^{\lambda t}f(x,v),e^{\lambda t}\phi(x))$ to
(\ref{vp}) with $\lambda>0$ satisfies%
\begin{equation}
\lambda f+v\cdot\nabla_{x}f-\nabla_{x}U_{0}\cdot\nabla_{v}f=f_{0}^{\prime
}v\cdot\nabla_{x}\phi. \label{linear-vlasov}%
\end{equation}
We define $[X(s;x,v),V(s;x,v)]$ as the trajectory of%

\begin{equation}
\left\{
\begin{array}
[c]{c}%
\frac{dX(s;x,v)}{ds}=V(s;x,v)\\
\frac{dV(s;x,v)}{ds}=-\nabla_{x}U_{0}%
\end{array}
\right.  \label{trajectory-ode}%
\end{equation}
such that $X(0;x,v)=x,~$and $V(0;x,v)=v.$ Notice that the particle energy $E$
is constant along the trajectory. Integrating along such a trajectory for
$-\infty\leq s\leq0$, we have%
\begin{align}
f(x,v)  &  =\int_{-\infty}^{0}e^{\lambda s}f_{0}^{\prime}(E)V(s;x,v)\cdot
\nabla_{x}\phi(X(s;x,v))ds\label{f-exp}\\
&  =f_{0}^{\prime}(E)\phi(x)-f_{0}^{\prime}(E)\int_{-\infty}^{0}\lambda
e^{\lambda s}\phi(X(s;x,v))ds.\nonumber
\end{align}

Plugging it back into the Poisson equation, we obtain an equation for $\phi$%
\[
-\Delta\phi+[4\pi\int f_{0}^{\prime}(E)dv]\phi-4\pi\int f_{0}^{\prime}%
(E)\int_{-\infty}^{0}\lambda e^{\lambda s}\phi(X(s;x,v))dsdv=0.
\]
We therefore define the operator $A_{\lambda}$ as
\[
A_{\lambda}\phi\equiv-\Delta\phi+[4\pi\int f_{0}^{\prime}(E)dv]\phi-4\pi\int
f_{0}^{\prime}(E)\int_{-\infty}^{0}\lambda e^{\lambda s}\phi(X(s;x,v))dsdv.
\]

\begin{lemma}
\label{lemma-operator}Assume that $f_{0}(E)$ has a bounded support in $x$ and
$v$ and $f_{0}^{\prime}$ is bounded. For any $\lambda>0$, the operator
$A_{\lambda}:H^{2}\rightarrow$ $L^{2}$ is self-adjoint with the essential
spectrum $[0,+\infty)\,.$
\end{lemma}

\begin{proof}
We denote
\[
K_{\lambda}\phi=-4\pi\lbrack\int f_{0}^{\prime}(E)dv]\phi+4\pi\int
f_{0}^{\prime}(E)\int_{-\infty}^{0}\lambda e^{\lambda s}\phi
(X(s;x,v))dsdv.\text{ }%
\]
Recall that $f_{0}\left(  x,v\right)  =f_{0}(E)$ has a compact support
$\subset S\subset\mathbb{R}_{x}^{3}\times\mathbb{R}_{v}^{3}$. We may assume
$S=S_{x}\times S_{v}$, both balls in $\mathbb{R}^{3}$. Let $\chi=\chi\left(
|x|\right)  $ be a smooth cut-off function for the spatial support of $f_{0}$
in the physical space $S_{x}$; that is, $\chi\equiv1$ on the spatial support
of $f_{0}$ and has compact support inside $S_{x}$. Let $M_{\chi}$ be the
operator of multiplication by $\chi$. Then $K_{\lambda}=K_{\lambda}M_{\chi
}=M_{\chi}K_{\lambda}=M_{\chi}K_{\lambda}M_{\chi}$. Indeed,
\[
f_{0}^{\prime}\left(  x,v\right)  =f_{0}^{\prime}\left(
X(s;x,v),V(s;x,v)\right)
\]
because of the invariance of $E$ under the flow. So
\begin{align}
\left(  K_{\lambda}\phi\right)  \left(  x\right)   &  =-4\pi\lbrack\int
f_{0}^{\prime}(E)dv]\phi+4\pi\int f_{0}^{\prime}(E)\int_{-\infty}^{0}\lambda
e^{\lambda s}\phi(X(s;x,v))dsdv\label{support}\\
&  =-4\pi\lbrack\int f_{0}^{\prime}(E)dv]\phi+4\pi\int\int_{-\infty}%
^{0}\lambda e^{\lambda s}\left(  f_{0}^{\prime}(E)\phi\right)
(X(s;x,v))dsdv\nonumber\\
&  =(M_{\chi}K_{\lambda}M_{\chi}\phi)(x).\nonumber
\end{align}
First we claim that
\[
\left\Vert K_{\lambda}\right\Vert _{L^{2}\rightarrow L^{2}}\leq8\pi\left\vert
\int\left\vert f_{0}^{\prime}(E)\right\vert dv\right\vert _{\infty}.
\]
Indeed, the $L^{2}$ norm for the first term in $K_{\lambda}$ is easily bounded
by $4\pi\left\vert \int f_{0}^{\prime}(E)dv\right\vert _{\infty}$. For the
second term, we have for any $\psi\in L^{2},$
\begin{align}
&  |\int_{-\infty}^{0}\int\int4\pi\lambda e^{\lambda s}f_{0}^{\prime}%
(E)\phi(X(s;x,v))dsdv\psi(x)dx|\label{estimate-quadratic}\\
&  \leq4\pi\int_{-\infty}^{0}\lambda e^{\lambda s}\left(  \int\int
|f_{0}^{\prime}(E)|\phi^{2}(X(s;x,v))dvdx\right)  ^{\frac{1}{2}}\left(
\int\int|f_{0}^{\prime}(E)|\psi^{2}(x)dvdx\right)  ^{\frac{1}{2}}ds\nonumber\\
&  =4\pi\int_{-\infty}^{0}\lambda e^{\lambda s}\left(  \int\int|f_{0}^{\prime
}(E)|\phi^{2}(x)dvdx\right)  ^{\frac{1}{2}}\left(  \int\int|f_{0}^{\prime
}(E)|\psi^{2}(x)dvdx\right)  ^{\frac{1}{2}}ds\nonumber\\
&  =4\pi\left(  \int\int|f_{0}^{\prime}(E)|\phi^{2}(x)dvdx\right)  ^{\frac
{1}{2}}\left(  \int\int|f_{0}^{\prime}(E)|\psi^{2}(x)dvdx\right)  ^{\frac
{1}{2}}\nonumber\\
&  \leq4\pi\left\vert \int\left\vert f_{0}^{\prime}(E)\right\vert
dv\right\vert _{\infty}\left\Vert \phi\right\Vert _{2}\left\Vert
\psi\right\Vert _{2}\text{. }\nonumber
\end{align}
Moreover, we have that $K_{\lambda}$ is symmetric Indeed, for fixed $s,$ by
making a change of variable $(y,w)\rightarrow(X(s;x,v),V(s;x,v)),$ so that
$(z,v)=(X(-s;y,w),V(-s;y,w)),$ we deduce that
\begin{align*}
&  \int\int4\pi f_{0}^{\prime}(E)\int_{-\infty}^{0}\lambda e^{\lambda s}%
\phi(X(s;x,v))dsdv\psi(x)dx\\
&  =\int_{-\infty}^{0}\lambda e^{\lambda s}\int\int4\pi f_{0}^{\prime}%
(E)\phi(y)\psi(X(-s;y,w))dydwds\\
&  =\int\int4\pi f_{0}^{\prime}(E)\int_{-\infty}^{0}\lambda e^{\lambda s}%
\psi(X(-s;y,-w))\phi(y)dydwds\\
&  =\int\int4\pi f_{0}^{\prime}(E)\int_{-\infty}^{0}\lambda e^{\lambda s}%
\psi(X(s;x,v))\phi(x)dvdxds.
\end{align*}
Here we have used the fact $[X(s;y,w),V(s;y,w)]=[X(-s;y,-w),-V(s;y,-w)]$ in
the last line$.$ Hence%
\[
(K_{\lambda}\phi,\psi)=(\phi,K_{\lambda}\psi).
\]
Since $K_{\lambda}=K_{\lambda}M_{\chi}$ and $M_{\chi}$ is compact from $H^{2}$
into $L^{2}$ space with support in $S_{x}$, so $K_{\lambda}$ is relatively
compact with respect to $-\Delta$. Thus by Kato-Relich and Weyl's Theorems,
$A_{\lambda}:H^{2}\rightarrow$ $L^{2}$ is self-adjoint and $\sigma
_{\text{ess}}(A_{\lambda})=\sigma_{\text{ess}}(-\Delta).$
\end{proof}

\begin{lemma}
\label{lemma-infy}Assume that $f_{0}^{\prime}(E)$ has a bounded support in $x$
and $v$ and $f_{0}^{\prime}$ is bounded. Let
\[
k(\lambda)=\inf_{\phi\in D(A_{\lambda}),||\phi||_{2}=1}(\phi,A_{\lambda}%
\phi),
\]
then $k(\lambda)$ is a continuous function of $\lambda$ when $\lambda>0$.
Moreover, there exists $0<\Lambda<\infty$ such that for $\lambda>\Lambda$
\begin{equation}
k(\lambda)\geq0. \label{limit-infy}%
\end{equation}

\end{lemma}

\begin{proof}
Fix $\lambda_{0}>0,$ $\phi\in D(A_{\lambda}),$ and $||\phi||_{2}=1.$ Then
\begin{align*}
k(\lambda_{0})  &  \leq(\phi,A_{\lambda_{0}}\phi)\\
&  \leq(\phi,A_{\lambda}\phi)+|(\phi,A_{\lambda_{0}}\phi)-(\phi,A_{\lambda
}\phi)|\\
&  \leq(\phi,A_{\lambda}\phi)+4\pi\int\int|f_{0}^{\prime}(E)|\int_{-\infty
}^{0}[\lambda e^{\lambda s}-\lambda_{0}e^{\lambda_{0}s}]\phi(X(s;x,v))\phi
(x)dsdvdx\\
&  \leq(\phi,A_{\lambda}\phi)+4\pi\int\int|f_{0}^{\prime}(E)|\int_{-\infty
}^{0}\int_{\lambda_{0}}^{\lambda}[\tilde{\lambda}|s|e^{\tilde{\lambda}%
s}+e^{\tilde{\lambda}s}]d\tilde{\lambda}\phi(X(s;x,v))\phi(x)dsdvdx\\
&  \leq(\phi,A_{\lambda}\phi)+C\int_{-\infty}^{0}\int_{\lambda_{0}}^{\lambda
}[\tilde{\lambda}|s|e^{\tilde{\lambda}s}+e^{\tilde{\lambda}s}]d\tilde{\lambda
}ds\\
&  \leq(\phi,A_{\lambda}\phi)+C|\ln\lambda-\ln\lambda_{0}|.
\end{align*}
We therefore deduce that by taking the infimum over all $\phi,$%
\[
k(\lambda_{0})\leq k(\lambda)+C|\ln\lambda-\ln\lambda_{0}|.
\]
Same argument also yields $k(\lambda)\leq k(\lambda_{0})+C|\ln\lambda
-\ln\lambda_{0}|.$Thus $\left\vert k(\lambda_{0})-k(\lambda)\right\vert \leq
C|\ln\lambda-\ln\lambda_{0}|$ and $k(\lambda)$ is continuous for $\lambda>0$.

To prove (\ref{limit-infy}), by (\ref{f-exp}), we recall from Sobolev's
inequality in $\mathbf{R}^{3}$%

\begin{align*}
|(K_{\lambda}\phi,\psi)|  &  =\left\vert \int\int4\pi f_{0}^{\prime
}(E)e^{\lambda s}\nabla\phi(X(s;x,v))V(s)dsdv\psi(x)dx\right\vert \\
&  \leq\int_{-\infty}^{0}e^{\lambda s}\left(  \int\int|\psi|^{2}|f_{0}%
^{\prime}(E)|dvdx\right)  ^{1/2}\cdot\\
&  \times\lbrack\int\int|\nabla\phi(X\left(  s\right)  )|^{2}|f_{0}^{\prime
}(E)||V\left(  s\right)  |^{2}dxdv]^{1/2}ds\\
&  =\int_{-\infty}^{0}e^{\lambda s}\left(  \int\int|\psi|^{2}|f_{0}^{\prime
}(E)|dvdx\right)  ^{1/2}\int\int v^{2}|\nabla\phi(x)|^{2}|f_{0}^{\prime
}(E)|dxdv]^{1/2}ds\\
&  \leq\frac{C}{\lambda}||\psi||_{6}||\nabla\phi||_{2}\leq\frac{C}{\lambda
}||\nabla\psi||_{2}||\nabla\phi||_{2},
\end{align*}
since $f_{0}$ has compact support. Therefore,
\[
(A_{\lambda}\phi,\phi)=||\nabla\phi||^{2}-(K_{\lambda}\phi,\phi)\geq
(1-\frac{C}{\lambda})||\nabla\phi||^{2}\geq0
\]
for $\lambda$ large.
\end{proof}

We now compute $\lim_{\lambda\rightarrow0+}A_{\lambda}$. We first consider the
case when the test function $\phi$ is spherically symmetric.

\begin{lemma}
\label{lemma-a0-radial}For spherically symmetric function $\phi(x)=\phi\left(
|x|\right)  ,$ we have
\begin{equation}%
\begin{split}
\lim_{\lambda\rightarrow0+}(A_{\lambda}\phi,\phi)  &  =(A_{0}\phi,\phi
)\equiv\int|\nabla\phi|^{2}dx+4\pi\int\int f_{0}^{\prime}(E)dv\phi^{2}dx\\
&  \ \ \ \ \ \ \ \ \ -32\pi^{3}\int_{\min U_{0}}^{E}\int_{0}^{\infty}%
f_{0}^{\prime}(E)\frac{\left(  \int_{r_{1}(E,L)}^{r_{2}(E,L)}\frac{\phi
dr}{\sqrt{2(E-U_{0}-L^{2}/2r^{2})}}\right)  ^{2}}{\int_{r_{1}(E,L)}%
^{r_{2}(E,L)}\frac{dr}{\sqrt{2(E-U_{0}-L^{2}/2r^{2})}}}dLdE\\
&  =\int|\nabla\phi|^{2}+32\pi^{3}\int f_{0}^{\prime}(E)\int_{r_{1}%
(E,L)}^{r_{2}(E,L)}(\phi-\bar{\phi})^{2}\frac{drdEdL}{\sqrt{2(E-U_{0}%
-L^{2}/2r^{2})}}.
\end{split}
\label{a0-radial}%
\end{equation}

\end{lemma}

\begin{proof}
Given the steady state $f_{0}(E)$, $U_{0}(|x|)$ and any radial function
$\phi\left(  \left\vert x\right\vert \right)  .$ To find the limit of
\begin{align}
(A_{\lambda}\phi,\phi)  &  =\int|\nabla\phi|^{2}dx+4\pi\int\int f_{0}^{\prime
}(E)dv\phi^{2}dx\label{exp-A_lambda}\\
&  -4\pi\int\int f_{0}^{\prime}(E)\left(  \int_{-\infty}^{0}\lambda e^{\lambda
s}\phi(X(s;x,v))ds\ \right)  \phi\left(  x\right)  dxdv,\nonumber
\end{align}
we study the following
\begin{equation}
\lim_{\lambda\rightarrow0+}\int_{-\infty}^{0}\lambda e^{\lambda s}%
\phi(X(s;x,v))ds. \label{limit-radial}%
\end{equation}
Note that we only need to study (\ref{limit-radial}) for points $\left(
x,v\right)  $ with $E=\frac{1}{2}|v|^{2}+U_{0}|\left(  x|\right)  <E_{0}$ and
$L=\left\vert x\times v\right\vert >0$, because in the third integral of
(\ref{exp-A_lambda}) $f_{0}^{\prime}(E)$ has support in $\left\{
E<E_{0}\right\}  $ and the set $\left\{  L=0\right\}  $ has a zero measure. We
recall the linearized Vlasov-Poisson system in the $r,v_{r},L$ coordinates
takes the form%
\begin{align*}
\partial_{t}f+v_{r}\partial_{r}f+\left(  \frac{L}{r^{3}}-\partial_{r}%
U_{0}\right)  \partial_{v_{r}}f  &  =\partial_{r}U_{f}\partial_{v_{r}}f_{0},\\
\partial_{rr}U_{f}+\frac{2}{r}\partial_{r}U_{f}  &  =4\pi\int fdv.
\end{align*}
For the corresponding linearized system, for points $\left(  x,v\right)  $
with $E<E_{0}$ and $L>0,$ the trajectory of $(X(s;x,v),V(s;x,v))$ in the
coordinate $(r,E,L)\,$ is a periodic motion described by the ODE (see
\cite{BT})
\begin{align*}
\frac{dr(s)}{ds}  &  =v_{r}(s),\\
\frac{dv_{r}(s)}{ds}  &  =-U_{0}^{\prime}(r)+\frac{L^{2}}{r^{3}}.
\end{align*}
with the period
\[
T\left(  E,L\right)  =2\int_{r_{1}(E,L)}^{r_{2}(E,L)}\frac{dr}{\sqrt
{2(E-U_{0}-L^{2}/2r^{2})}},
\]
where $0<r_{1}(E,L)\leq r_{2}(E,L)<+\infty$ are zeros of $E-U_{0}-L^{2}%
/2r^{2}.$So by Lin's lemma in [\cite{lin01}]\textbf{,}
\[
\lim_{\lambda\rightarrow0}\int_{-\infty}^{0}\lambda e^{\lambda s}%
\phi(X(s;x,v))ds=\frac{1}{T}\int_{0}^{T}\phi(X(s;x,v))ds.
\]
Since $\phi(X(s;x,v)=\phi(r(s)),\ $a change of variable from $s\rightarrow
r(s)$ leads to%
\[
\int_{0}^{T}\phi(X(s;x,v))ds=2\int_{r_{1}}^{r_{2}}\frac{\phi(r)dr}%
{\sqrt{2(E-U_{0}-L^{2}/2r^{2})}}.
\]
For any function $g(r,E,L),$ we define its trajectory average as
\[
\bar{g}(E,L)\equiv\frac{\int_{r_{1}(E,L)}^{r_{2}(E,L)}\frac{g(r,E,L)dr}%
{\sqrt{2(E-U_{0}-L^{2}/2r^{2})}}}{\int_{r_{1}(E,L)}^{r_{2}(E,L)}\frac
{dr}{\sqrt{2(E-U_{0}-L^{2}/2r^{2})}}}.
\]
Then
\[
\lim_{\lambda\rightarrow0+}\int_{-\infty}^{0}\lambda e^{\lambda s}%
\phi(X(s;x,v))ds=2\int_{r_{1}}^{r_{2}}\frac{\phi(r)dr}{\sqrt{2(E-U_{0}%
-L^{2}/2r^{2})}}/T\left(  E,L\right)  =\bar{\phi}\left(  E,L\right)
\]
and the integrand in third term of (\ref{exp-A_lambda}) converges pointwise to
$f_{0}^{\prime}(E)\bar{\phi}\phi$. Thus by the dominated convergence theorem,
we have
\begin{align*}
\lim_{\lambda\rightarrow0+}(A_{\lambda}\phi,\phi)  &  =\int|\nabla\phi
|^{2}dx+4\pi\int\int f_{0}^{\prime}(E)\phi^{2}dxdv-4\pi\int\int f_{0}^{\prime
}(E)\bar{\phi}\phi\ dxdv\\
&  =\int|\nabla\phi|^{2}dx+4\pi\int\int f_{0}^{\prime}(E)\phi^{2}dxdv\\
\ \ \ \ \ \ \  &  \ \ \ \ \ \ -32\pi^{3}\int_{\min U_{0}}^{E}\int_{0}^{\infty
}f_{0}^{\prime}(E)\int_{r_{1}(E,L)}^{r_{2}(E,L)}\bar{\phi}\left(  E,L\right)
\phi\left(  r\right)  \frac{drdEdL}{\sqrt{2(E-U_{0}-L/2r^{2})}}\\
&  =\int|\nabla\phi|^{2}dx+4\pi\int\int f_{0}^{\prime}(E)\phi^{2}dxdv\\
&  \ \ \ \ \ \ \ -32\pi^{3}\int_{\min U_{0}}^{E}\int_{0}^{\infty}f_{0}%
^{\prime}(E)\frac{\left(  \int_{r_{1}(E,L)}^{r_{2}(E,L)}\frac{\phi dr}%
{\sqrt{2(E-U_{0}-L/2r^{2})}}\right)  ^{2}}{\int_{r_{1}(E,L)}^{r_{2}(E,L)}%
\frac{dr}{\sqrt{2(E-U_{0}-L/2r^{2})}}}dEdL\\
&  =\int|\nabla\phi|^{2}+32\pi^{3}\int f_{0}^{\prime}(E)\int_{r_{1}%
(E,L)}^{r_{2}(E,L)}(\phi-\bar{\phi})^{2}\frac{drdEdL}{\sqrt{2(E-U_{0}%
-L/2r^{2})}}.
\end{align*}
This finishes the proof of the lemma.
\end{proof}

To compute $\lim_{\lambda\rightarrow0+}(A_{\lambda}\phi,\phi)$ for more
general test function $\phi,$ we use the following ergodic lemma which is a
direct generalization of the result in \cite{lw-linear}.

\begin{lemma}
\label{lemma-ergodic}Consider the solution $\left(  P\left(  s;p,q\right)
,Q\left(  s;p,q\right)  \right)  $ to be the solution of a Hamiltonian system
\begin{align*}
\dot{P}  &  =\partial_{q}H\left(  P,Q\right) \\
\dot{Q}  &  =-\partial_{p}H\left(  P,Q\right)
\end{align*}
with $\left(  P\left(  0\right)  ,Q\left(  0\right)  \right)  =\left(
p,q\right)  \in\mathbf{R}^{n}\times\mathbf{R}^{n}$. Denote
\[
\mathcal{Q}^{\lambda}m=\int_{-\infty}^{0}\lambda e^{\lambda s}m\left(
P\left(  s\right)  ,Q\left(  s\right)  \right)  ds.
\]
Then for any $m\left(  p,q\right)  \in L^{2}\left(  \mathbf{R}^{n}%
\times\mathbf{R}^{n}\right)  $, we have $\mathcal{Q}^{\lambda}m\rightarrow
\mathcal{P}m$ strongly in $L^{2}\left(  \mathbf{R}^{n}\times\mathbf{R}%
^{n}\right)  $. Here $\mathcal{P}$ is the projection operator of $L^{2}\left(
\mathbf{R}^{n}\times\mathbf{R}^{n}\right)  $ to the kernel of the transport
operator $D=\partial_{q}H\partial_{p}-\partial_{p}H\partial_{q}$ and
$\mathcal{P}m$ is the phase space average of $m$ in the set traced by the trajectory.
\end{lemma}

\begin{proof}
Denote $U\left(  s\right)  :L^{2}\left(  \mathbf{R}^{n}\times\mathbf{R}%
^{n}\right)  \rightarrow L^{2}\,\left(  \mathbf{R}^{n}\times\mathbf{R}%
^{n}\right)  $\ to be the unitary semigroup $U\left(  s\right)  m=m\left(
P\left(  s\right)  ,Q\left(  s\right)  \right)  $. By Stone Theorem
(\cite{yosida}), $U\left(  s\right)  $ is generated by $iR=D$, where $R=-iD$
is self-adjoint and
\[
U\left(  s\right)  =\int_{-\infty}^{+\infty}e^{i\alpha s}dM_{\alpha}%
\]
where $\left\{  M_{\alpha};\alpha\in\mathbf{R}^{1}\right\}  $ is spectral
measure of $R$. So
\[
\int_{-\infty}^{0}\lambda e^{\lambda s}m(P(s),Q(s))ds=\int_{-\infty}%
^{0}\lambda e^{\lambda s}\int_{\mathbb{R}}e^{i\alpha s}dM_{\alpha}%
m\ ds=\int_{\mathbb{R}}\frac{\lambda}{\lambda+i\alpha}dM_{\alpha}m.
\]
On the other hand, the projection is $\mathcal{P}=M_{\{0\}}=\int_{\mathbb{R}%
}\xi dM_{\alpha}$ where $\xi(\alpha)=0$ for $\alpha\neq0$ and $\xi(0)=1$.
Therefore
\[
\left\Vert \int_{-\infty}^{0}\lambda e^{\lambda s}m(P(s),Q(s))ds-\mathcal{P}%
m\right\Vert _{\mathbf{L}^{2}}^{2}=\int_{\mathbb{R}}\left\vert \frac{\lambda
}{\lambda+i\alpha}-\xi(\alpha)\right\vert ^{2}d\Vert M_{\alpha}m\Vert
_{\mathbf{L}^{2}}^{2}%
\]
by orthogonality of the spectral projections. By the dominated convergence
theorem this expression tends to $0$ as $\lambda\rightarrow0+$, as we wished
to prove. The explaination of $\mathcal{P}m$ as the phase space average of $m$
is in our remark below.
\end{proof}

\begin{remark}
Since $\int_{-\infty}^{0}\lambda e^{\lambda s}ds=1$, the function%
\begin{equation}
\left(  \mathcal{Q}^{\lambda}m\right)  \left(  x,v\right)  =\int_{-\infty}%
^{0}\lambda e^{\lambda s}m\left(  P(s),Q(s)\right)  ds\label{time-avera-wei}%
\end{equation}
is a weighted time average of the observable $m$ along the particle
trajectory. By the same proof of Lemma \ref{lemma-ergodic}, we have
\begin{equation}
\lim_{T\rightarrow\infty}\frac{1}{T}\int_{0}^{T}m\left(  P(s),Q(s)\right)
ds=\mathcal{P}m.\label{time-average}%
\end{equation}
But from the standard ergodic theory (\cite{arnold-ergodic}) of Hamiltonian
systems, the limit of the above time average in (\ref{time-average}) equals
the phase space average of $m$ in the set traced by the trajectory. Thus
$\mathcal{P}m$ has the meaning of the phase space average of $m$ and Lemma
\ref{lemma-ergodic} states that the limit of the weighted time average
(\ref{time-avera-wei}) yields the same phase space average. In particular, if
the particle motion is ergodic in the invariant set $S_{I}$ determined by the
invariants $E_{1},\cdots,I_{k}$, and if $d\sigma_{I}$ denotes the induced
measure of $\mathbf{R}^{n}\times\mathbf{R}^{n}$ on $S_{I}$, then
\begin{equation}
\mathcal{P}m=\frac{1}{\sigma_{I}\left(  S_{I}\right)  }\int_{S_{I}}m\left(
p,q\right)  d\sigma_{I}\left(  p,q\right)  .\label{ergodic}%
\end{equation}
For integral systems, using action angle variables $\left(  J_{1},\cdots
,J_{n};\varphi_{1},\cdots,\varphi_{n}\right)  $ we have
\begin{equation}
\left(  \mathcal{P}m\right)  \left(  J_{1},\cdots,J_{n}\right)  =\left(
2\pi\right)  ^{-n}\int_{0}^{2\pi}\cdots\int_{0}^{2\pi}m\left(  J_{1}%
,\cdots,J_{n},\varphi_{1},\cdots,\varphi_{n}\right)  d\varphi_{1},\cdots
d\varphi_{n}\label{ergodic-integrable}%
\end{equation}
for the generic case with independent frequencies (see \cite{arnold78}).
\end{remark}

Recall the weighted $L^{2}$ space $L_{\left\vert f_{0}^{\prime}\right\vert
}^{2}$ in (\ref{weight}). Then $U\left(  s\right)  :L_{\left\vert
f_{0}^{\prime}\right\vert }^{2}\rightarrow L_{\left\vert f_{0}^{\prime
}\right\vert }^{2}$ defined by $U\left(  s\right)  m=m\left(
X(s;x,v),V(s;x,v)\right)  $ is an unitary group, where $\left(
X(s;x,v),V(s;x,v)\right)  $ is the particle trajectory (\ref{trajectory-ode}).
The generator of $U\left(  s\right)  $ is $D=v\cdot\partial_{x}-\nabla
_{x}U_{0}\cdot\nabla_{v}$ and $R=-iD$ is self-adjoint by Stone Theorem. By the
same proof, Lemma \ref{lemma-ergodic} is still valid in $L_{\left\vert
f_{0}^{\prime}\right\vert }^{2}$. In particular, for any $\phi\left(
x\right)  \in L^{2}\left(  \mathbf{R}^{3}\right)  $ we have
\begin{equation}
\int_{-\infty}^{0}\lambda e^{\lambda s}\phi(X(s;x,v))ds\rightarrow
\mathcal{P}\phi\label{ergodic-galaxy}%
\end{equation}
in $L_{\left\vert f_{0}^{\prime}\right\vert }^{2}$, where $\mathcal{P}$ is the
projector of $L_{\left\vert f_{0}^{\prime}\right\vert }^{2}$ to $\ker D$.

Now we derive an explicit formula for the above limit $\mathcal{P}\phi$. Note
that as in the proof of lemma \ref{lemma-a0-radial}, we only need to derive
the formula of $\mathcal{P}\phi$ for points $\left(  x,v\right)  $ with
$E<E_{0}$ and $L>0$. Since $U_{0}\left(  x\right)  =U_{0}\left(  r\right)  $,
the particle motion (\ref{trajectory-ode}) in such a center field is
integrable and has been well studied (see e.g. \cite{BT}, \cite{arnold78}).
For particles with energy $E<E_{0}<0$, $L>0$ and momentum $\vec{L}=x\times v$,
the particle orbit is a rosette in the annulus
\[
A_{E,L}=\left\{  r_{1}(E,L)\leq r\leq r_{2}(E,L)\right\}  =\left\{
E-U_{0}-L^{2}/2r^{2}\geq0\right\}  ,
\]
lying on the orbital plane perpendicular to $\vec{L}$. So we can consider the
particle motion to be planar. For such case, the action-angle variables are as
follows (see e.g. \cite{lb94}): the actions variables are%

\[
J_{r}=\frac{2\pi}{T\left(  E,L\right)  },\text{ \ \ \ \ \ }J_{\theta}=L,
\]
where
\[
T\left(  E,L\right)  =2\int_{r_{1}(E,L)}^{r_{2}(E,L)}\frac{dr}{\sqrt
{2(E-U_{0}-L^{2}/2r^{2})}}.
\]
is the radial period, the angle variable $\varphi_{r}$ is determined by
\[
d\varphi_{r}=\frac{2\pi}{T\left(  E,L\right)  }\frac{dr}{\sqrt{2(E-U_{0}%
-L^{2}/2r^{2})}}%
\]
and $\varphi_{\theta}=\theta-\Delta\theta$ where
\[
d\left(  \Delta\theta\right)  =\frac{Lr^{-2}-\Omega_{\theta}}{\sqrt
{2(E-U_{0}-L^{2}/2r^{2})}}dr
\]
and
\[
\Omega_{\theta}\left(  E,L\right)  =\frac{1}{T\left(  E,L\right)  }\int
_{r_{1}(E,L)}^{r_{2}(E,L)}\frac{L}{r^{2}\sqrt{2(E-U_{0}-L^{2}/2r^{2})}}dr
\]
is the average angular velocity. For any function $\phi\left(  x\right)  \in
H^{2}\left(  \mathbf{R}^{3}\right)  $, we denote $\phi_{\vec{L}}\left(
r,\theta\right)  $ to be the restriction of $\phi$ in the orbital plane
perpendicular to $\vec{L}$. Then by (\ref{ergodic-integrable}), for the
generic case when the radial and angular frequencies are independent, we have
\begin{align}
\left(  \mathcal{P}\phi\right)  \left(  E,\vec{L}\right)   &  =\left(
2\pi\right)  ^{-2}\int_{0}^{2\pi}\int_{0}^{2\pi}\phi_{\vec{L}}d\varphi
_{\theta}d\varphi_{r}\label{ergodic-spherical}\\
&  =\frac{1}{\pi T\left(  E,L\right)  }\int_{r_{1}(E,L)}^{r_{2}(E,L)}\int
_{0}^{2\pi}\frac{\phi_{\vec{L}}\left(  r,\theta\right)  d\theta dr}%
{\sqrt{2(E-U_{0}-L^{2}/2r^{2})}}.\nonumber
\end{align}
In particular, for a spherically symmetric function $\phi=\phi\left(
r\right)  $, we recover
\begin{equation}
\left(  \mathcal{P}\phi\right)  \left(  E,L\right)  =\frac{2}{T\left(
E,L\right)  }\int_{r_{1}(E,L)}^{r_{2}(E,L)}\frac{\phi(r)dr}{\sqrt
{2(E-U_{0}-L^{2}/2r^{2})}}. \label{ergodic-radial}%
\end{equation}
We thus conclude the following

\begin{lemma}
\label{lemma-limit0}Assume that $f_{0}(E)$ has a bounded support in $x$ and
$v$ and $f_{0}^{\prime}$ is bounded. For any $\phi\in H^{1}\left(
\mathbf{R}^{3}\right)  $, we have%
\begin{align}
\lim_{\lambda\rightarrow0+}(A_{\lambda}\phi,\phi)  &  =(A_{0}\phi
,\phi)\label{A0-general}\\
&  =\int|\nabla\phi|^{2}dx+4\pi\int\int f_{0}^{\prime}(E)dv\phi^{2}dx-4\pi
\int\int f_{0}^{\prime}(E)\left(  \mathcal{P}\phi\right)  ^{2}dxdv\nonumber\\
&  =\int|\nabla\phi|^{2}dx+4\pi\int\int f_{0}^{\prime}(E)\left(
\phi-\mathcal{P}\phi\right)  ^{2}dxdv\nonumber
\end{align}
where $\mathcal{P}$ is the projector of $L_{\left\vert f_{0}^{\prime
}\right\vert }^{2}$ to $\ker D$ and more explicitly $\mathcal{P}\phi$ is given
by (\ref{ergodic-spherical}). The limiting operator $A_{0}$ is
\begin{equation}
A_{0}\phi=-\Delta\phi+[4\pi\int f_{0}^{\prime}(E)dv]\phi-4\pi\int
f_{0}^{\prime}(E)\mathcal{P}\phi dv. \label{operator-A0}%
\end{equation}

\end{lemma}

Now we give the proof of the instability criterion.

\begin{proof}
[Proof of Theorem \ref{theorem-insta}]We define
\[
\lambda_{\ast}=\sup_{k(\lambda)<0}\lambda.
\]
By Lemmas \ref{lemma-operator} and \ref{lemma-limit0}, we deduce that
\[
-\infty<\lambda_{\ast}\leq\Lambda<\infty.
\]
Therefore, by the continuity of $k(\lambda),$ we have
\[
k(\lambda_{\ast})=0.
\]
Hence, there exists an increasing sequence of $\lambda_{n}<\lambda
_{n+1}<\lambda_{\ast}$ so that $\lambda_{n}\rightarrow\lambda_{\ast}$,
$k_{n}\equiv k(\lambda_{n})<0,$ and
\[
k_{n}\rightarrow k(\lambda_{\ast})=0.
\]
Therefore, $k_{n}$ are negative eigenvalues. By Lemma \ref{lemma-infy}, we get
a sequence $\phi_{n}\in H^{2}$ such that
\begin{equation}
A_{\lambda_{n}}\phi_{n}=k_{n}\phi_{n} \label{eqn-eigenvalue}%
\end{equation}
with $k_{n}<0$, $k_{n}\rightarrow0$ and $\lambda_{n}\rightarrow\lambda_{0}>0$,
as $n\rightarrow\infty$. Recall $\chi$ the cutoff function of the support of
$f_{0}(E)$ such that $\chi\equiv1$ for $f_{0}(E)>0.$ We claim that $\chi
\phi_{n}$ is a nonzero function for any $n$. Suppose otherwise, $\chi\phi
_{n}\equiv0$, then from the equation (\ref{eqn-eigenvalue}) we have $\left(
-\Delta-k_{n}\right)  \phi_{n}=0$ which implies that $\phi_{n}=0$, a
contradiction$.$Thus we can normalize $\phi_{n}$ by $\left\Vert \chi\phi
_{n}\right\Vert _{2}=1$. Taking inner product of (\ref{eqn-eigenvalue}) with
$\phi_{n}$ and integrating by parts, we have%
\begin{align*}
\left\Vert \bigtriangledown\phi_{n}\right\Vert _{2}^{2}  &  \leq-4\pi\int\int
f_{0}^{\prime}(E)\phi_{n}^{2}\ dvdx+\int\int4\pi f_{0}^{\prime}(E)\int
_{-\infty}^{0}\lambda_{n}e^{\lambda_{n}s}\phi_{n}(X(s;x,v))ds\phi_{n}\left(
x\right)  dx\\
&  =-4\pi\int\int f_{0}^{\prime}(E)\left(  \chi\phi_{n}\right)  ^{2}\ dvdx\\
&  +\int\int4\pi f_{0}^{\prime}(E)\int_{-\infty}^{0}\lambda_{n}e^{\lambda
_{n}s}\left(  \chi\phi_{n}\right)  (X(s;x,v))ds\left(  \chi\phi_{n}\right)
\left(  x\right)  dx\\
&  \leq8\pi\left\vert \int f_{0}^{\prime}(E)dv\right\vert _{\infty}\left\Vert
\chi\phi_{n}\right\Vert _{2}^{2}.
\end{align*}
Here in the second equality above, we use the fact $\chi=1\ $on the support of
$f_{0}^{\prime}(E)\ $($f_{0}(E)$) and that $\left(  \chi\phi_{n}\right)
(X(s;x,v))=$ $\phi_{n}(X(s;x,v)\chi\ $due to the invariance of the support
under the trajectory flow, as in (\ref{support}). In the last inequality, we
use the same estimate as in (\ref{estimate-quadratic}). Thus,
\[
\sup_{n}||\phi_{n}||_{L^{6}}\leq C\sup_{n}\left\Vert \bigtriangledown\phi
_{n}\right\Vert _{2}<C^{\prime},
\]
for some constant $C^{\prime}$ independent of $n$. Then there exists $\phi\in
L^{6}$ and $\nabla\phi\in L^{2}$ such that
\[
\phi_{n}\rightarrow\phi\text{ weakly in }L^{6}\text{, }\ \ \ \ \ \ \text{and
}\nabla\phi_{n}\rightarrow\nabla\phi\text{ weakly in }L^{2}.
\]
This implies that $\chi\phi_{n}\rightarrow\chi\phi$ strongly in $L^{2}$.
Therefore $\left\Vert \chi\phi\right\Vert _{2}=1$ and thus $\phi\neq0.$ It is
easy to show that $\phi$ is a weak solution of $A_{\lambda_{0}}\phi=0$ or
\begin{equation}
-\Delta\phi=-[4\pi\int f_{0}^{\prime}(E)dv]\phi+4\pi f_{0}^{\prime}%
(E)\int_{-\infty}^{0}\lambda_{0}e^{\lambda_{0}s}\phi(X(s;x,v))dsdv=\rho.
\label{poisson}%
\end{equation}
We have that
\begin{align*}
\int\rho dx  &  =-4\pi\int\int f_{0}^{\prime}(E)\phi\left(  x\right)
dxdv+\int_{-\infty}^{0}\lambda_{0}e^{\lambda_{0}s}\int\int4\pi f_{0}^{\prime
}(E)\phi(X(s;x,v))dxdvds\\
&  =-4\pi\int\int f_{0}^{\prime}(E)\phi\left(  x\right)  dxdv+\int_{-\infty
}^{0}\lambda_{0}e^{\lambda_{0}s}\int\int4\pi f_{0}^{\prime}(E)\phi(x)dxdvds=0
\end{align*}
and by (\ref{poisson}) $\rho$ has compact support in $S_{x}$, the $x-$support
of $f_{0}(E).$ Therefore from the formula $\phi\left(  x\right)  =\int
\frac{\rho\left(  y\right)  }{\left\vert x-y\right\vert }dy$, we have
\[
\phi\left(  x\right)  =\int\frac{\rho\left(  y\right)  }{\left\vert
x-y\right\vert }dy=\int\frac{\rho\left(  y\right)  }{\left\vert x-y\right\vert
}dy-\int\frac{\rho\left(  y\right)  }{\left\vert x\right\vert }dy=O\left(
\left\vert x\right\vert ^{-2}\right)  ,
\]
for $x$ large, and thus $\phi\in L^{2}$. By elliptic regularity, $\phi\in
H^{2}$. We define $f\left(  x,v\right)  $ by (\ref{f-exp}), then $f\in
L^{\infty}$ with the compact support in $S$. Now we show that $e^{\lambda
_{0}t}[f,\phi]$ is a weak solution to the linearized Vlasov-Poisson system.
Since $\phi$ satisfies the Poisson equation (\ref{poisson}), we only need to
show that $f$ satisfies the linearized Vlasov equation (\ref{linear-vlasov})
weakly. For that, we take any $g\in C_{c}^{1}\left(  \mathbb{R}^{3}%
\times\mathbb{R}^{3}\right)  ,$ and
\begin{align*}
&  \iint_{\mathbb{R}^{3}\times\mathbb{R}^{3}}\left(  Dg\right)  fdxdv\\
&  =\iint_{\mathbb{R}^{3}\times\mathbb{R}^{3}}\left(  Dg\right)  \left(
f_{0}^{\prime}(E)\phi(x)\right)  dxdv-\iint_{\mathbb{R}^{3}\times
\mathbb{R}^{3}}\left(  Dg\right)  f_{0}^{\prime}(E)\int_{-\infty}^{0}%
\lambda_{0}e^{\lambda_{0}s}\phi(X(s;x,v))dsdxdv\\
&  =I+II.
\end{align*}
Since $D$ is skew-adjoint, the first term is
\[
I=-\iint_{\mathbb{R}^{3}\times\mathbb{R}^{3}}gD\left(  f_{0}^{\prime}%
(E)\phi\right)  dxdv=-\iint_{\mathbb{R}^{3}\times\mathbb{R}^{3}}f_{0}^{\prime
}(E)gD\phi dxdv.
\]
For the second term,
\begin{align*}
II  &  =-\int_{-\infty}^{0}\lambda_{0}e^{\lambda_{0}s}\iint_{\mathbb{R}%
^{3}\times\mathbb{R}^{3}}f_{0}^{\prime}(E)\ Dg(x,v)\ \phi\left(
X(s;x,v)\right)  dxdvds\\
&  =-\int_{-\infty}^{0}\lambda_{0}e^{\lambda_{0}s}\iint_{\mathbb{R}^{3}%
\times\mathbb{R}^{3}}f_{0}^{\prime}(E)\left(  Dg\right)  \left(
X(-s),V(-s)\right)  \phi\left(  x\right)  dxdvds\\
&  =-\iint_{\mathbb{R}^{3}\times\mathbb{R}^{3}}f_{0}^{\prime}(E)\int_{-\infty
}^{0}\lambda_{0}e^{\lambda_{0}s}\left(  -\frac{d}{ds}g\left(
X(-s),V(-s)\right)  \right)  ds\ \phi\left(  x\right)  dxdv\\
&  =\iint_{\mathbb{R}^{3}\times\mathbb{R}^{3}}f_{0}^{\prime}(E)\left\{
\lambda_{0}g\left(  x,v\right)  -\int_{-\infty}^{0}\lambda_{0}^{2}%
e^{\lambda_{0}s}g\left(  X(-s),V(-s)\right)  ds\right\}  \phi\left(  x\right)
dxdv\\
&  =\iint_{\mathbb{R}^{3}\times\mathbb{R}^{3}}\left\{  f_{0}^{\prime
}(E)\lambda_{0}\phi\left(  x\right)  -f_{0}^{\prime}(E)\int_{-\infty}%
^{0}\lambda_{0}^{2}e^{\lambda_{0}s}\phi\left(  X(s),V(s)\right)  ds\right\}
g\left(  x,v\right)  dxdv\\
&  =\lambda_{0}\iint_{\mathbb{R}^{3}\times\mathbb{R}^{3}}\left\{
f_{0}^{\prime}(E)\phi\left(  x\right)  -f_{0}^{\prime}(E)\int_{-\infty}%
^{0}\lambda_{0}e^{\lambda_{0}s}\phi\left(  X(s),V(s)\right)  ds\right\}
g\ dxdv\\
&  =.\lambda_{0}\iint_{\mathbb{R}^{3}\times\mathbb{R}^{3}}fgdxdv.
\end{align*}
Thus we have
\[
\iint_{\mathbb{R}^{3}\times\mathbb{R}^{3}}\left(  Dg\right)  fdxdv=\iint
_{\mathbb{R}^{3}\times\mathbb{R}^{3}}\left(  \lambda_{0}f-f_{0}^{\prime
}(E)D\phi\right)  gdxdv
\]
which implies that $f$ is a weak solution to the linearized Vlasov equation
\[
\lambda_{0}f+Df=f_{0}^{\prime}\left(  E\right)  v\cdot\nabla_{x}\phi.
\]

\end{proof}

\begin{remark}
\label{remark anistropic}Consider an anisotropic spherical galaxy with
$f_{0}\left(  x,v\right)  =f_{0}\left(  E,L^{2}\right)  $. For a radial
symmetric growing mode $e^{\lambda t}\left(  \phi,f\right)  $ with $\phi
=\phi\left(  \left\vert x\right\vert \right)  $ and $f=f\left(  \left\vert
x\right\vert ,E,L^{2}\right)  $. The linearized Vlasov equation (\ref{lvp})
becomes
\begin{align*}
&  \ \ \ \ \ \lambda f+v\cdot\nabla_{x}f-\nabla_{x}U_{0}\cdot\nabla_{v}f\\
&  =\nabla_{x}\phi\cdot\nabla_{v}f_{0}=\nabla_{x}\phi\cdot\left(
\frac{\partial f_{0}}{\partial E}v+\frac{\partial f_{0}}{\partial L^{2}}%
\nabla_{v}\left(  \left\vert x\times v\right\vert ^{2}\right)  \right) \\
&  =\phi^{\prime}\left(  \left\vert x\right\vert \right)  \frac{x}{\left\vert
x\right\vert }\cdot\left(  \frac{\partial f_{0}}{\partial E}v+2\frac{\partial
f_{0}}{\partial L^{2}}\left[  \left(  x\times v\right)  \times x\right]
\right)  =\frac{\partial f_{0}}{\partial E}v\cdot\nabla_{x}\phi,
\end{align*}
which is of the same form as in the isotropic case (\ref{limit-radial}). So by
the same proof of Theorem \ref{theorem-insta}, we also get an instability
criterion for radial perturbations of anisotropic galaxy, in terms of the
quadratic form (\ref{a0-radial}) with $f_{0}^{\prime}(E)$ being replaced by
$\frac{\partial f_{0}}{\partial E}$.
\end{remark}

\section{Nonlinear Stability of the King's Model}

In the second half of the article, we investigate the nonlinear stability of
the King model (\ref{king}). We first establish:

\begin{lemma}
\label{lemma-ao}Consider spherical models $f_{0}=f_{0}\left(  E\right)  $ with
$f_{0}^{\prime}<0.$ The operator $A_{0}:H_{r}^{2}\rightarrow L_{r}^{2}$
\[
A_{0}\phi=-\Delta\phi+[4\pi\int f_{0}^{\prime}dv]\phi-4\pi\int f_{0}^{\prime
}\mathcal{P}\phi dv
\]
is positive, where $H_{r}^{2}$ and $L_{r}^{2}$ are spherically symmetric
subspaces of $H^{2}$ and $L^{2}$, and the projection $\mathcal{P}\phi$ is
defined by (\ref{ergodic-radial}). Moreover, for $\phi\in H_{r}^{2}$ we have
\begin{equation}
\left(  A_{0}\phi,\phi\right)  \geq\varepsilon\left(  \left\vert \nabla
\phi\right\vert _{2}^{2}+\left\vert \phi\right\vert _{2}^{2}\right)
\label{estimate-a0}%
\end{equation}
for some constant $\varepsilon>0$.
\end{lemma}

\begin{proof}
Define $k_{0}=\inf\left(  A_{0}\phi,\phi\right)  /\left(  \phi,\phi\right)
.$We want to show that $k_{0}>0$. First, by using the compact embedding of
$H_{r}^{2}\hookrightarrow L_{r}^{2}$ it is easy to show that the minimum can
be obtained and $k_{0}$ is the lowest eigenvalue. Let $A_{0}\phi_{0}=k_{0}%
\phi_{0}$ with $\phi_{0}\in H_{r}^{2}$ and $\left\Vert \phi_{0}\right\Vert
_{2}=1$. The fact that $k_{0}\geq0$ follows immediately from Theorem
\ref{theorem-insta} and the nonexistence of radial modes (\cite{dbf73},
\cite{KS85}) for monotone spherical models. The proof of $k_{0}>0$ is more
delicate. For that, we relate the quadratic form $\left(  A_{0}\phi
,\phi\right)  $ to the Antonov functional (\ref{Antonov}). We define
$D=v\cdot\partial_{x}-\nabla_{x}U_{0}\cdot\nabla_{v}$ to be the generator of
the unitary group $U\left(  s\right)  $:$L_{\left\vert f_{0}^{\prime
}\right\vert }^{2,r}\rightarrow L_{\left\vert f_{0}^{\prime}\right\vert
}^{2,r}$ defined by $U\left(  s\right)  m=m\left(  X(s;x,v),V(s;x,v)\right)
.$ Here $L_{\left\vert f_{0}^{\prime}\right\vert }^{2,r}$ is the spherically
symmetric subspace of $L_{\left\vert f_{0}^{\prime}\right\vert }^{2}$, which
is preserved under the flow mapping $U\left(  s\right)  $. By the definition
of $\mathcal{P}\phi$, we have $\phi_{0}-\mathcal{P}\phi_{0}\perp\ker D$. By
Stone theorem $iD$ is self-adjoint and in particular $D$ is closed. Therefore
by the closed range theorem (\cite{yosida}), we have $\left(  \ker D\right)
^{\perp}=R\left(  D\right)  $ , where $R\left(  D\right)  $ is the range of
$D$. So there exists $h\in L_{\left\vert f_{0}^{\prime}\right\vert }^{2,r}$
such that $Dh=\phi_{0}-\mathcal{P}\phi_{0}$. Moreover, since $\phi
_{0}-\mathcal{P}\phi_{0}$ is even in $v$ and the operator $D$ reverses the
parity in $v$, the function $h$ is odd in $v$. Define $f^{-}=f_{0}^{\prime}h.$
We have
\begin{align*}
k_{0}  &  =\left(  A_{0}\phi_{0},\phi_{0}\right)  =\int\left\vert \nabla
\phi_{0}\right\vert ^{2}dx+4\pi\int\int f_{0}^{\prime}\left(  \phi
_{0}-\mathcal{P}\phi_{0}\right)  ^{2}dxdv\\
&  =\int\left\vert \nabla\phi_{0}\right\vert ^{2}dx-8\pi\int\int\left\vert
f_{0}^{\prime}\right\vert \left(  \phi_{0}-\mathcal{P}\phi_{0}\right)
\phi_{0}dxdv\\
&  \ \ \ \ \ +4\pi\int\int\left\vert f_{0}^{\prime}\right\vert \left(
\phi_{0}-\mathcal{P}\phi_{0}\right)  ^{2}dxdv\\
&  =4\pi\left(  \int\int\frac{\left\vert Df^{-}\right\vert ^{2}}{\left\vert
f_{0}^{\prime}\right\vert }dxdv+2\int\phi_{0}\int Df^{-}dvdx+\frac{1}{4\pi
}\int\left\vert \nabla\phi_{0}\right\vert ^{2}dx\right) \\
&  =4\pi\left(  \int\int\frac{\left\vert Df^{-}\right\vert ^{2}}{\left\vert
f_{0}^{\prime}\right\vert }dxdv+\frac{1}{2\pi}\int\phi_{0}\Delta\phi
^{-}dx+\frac{1}{4\pi}\int\left\vert \nabla\phi_{0}\right\vert ^{2}dx\right) \\
&  =4\pi\left(  \int\int\frac{\left\vert Df^{-}\right\vert ^{2}}{\left\vert
f_{0}^{\prime}\right\vert }dxdv+\frac{1}{4\pi}\int\left(  \left\vert
\nabla\phi_{0}\right\vert ^{2}-2\nabla\phi_{0}\cdot\nabla\phi^{-}\right)
dx\right) \\
&  \geq4\pi\left(  \int\int\frac{\left\vert Df^{-}\right\vert ^{2}}{\left\vert
f_{0}^{\prime}\right\vert }dxdv-\frac{1}{4\pi}\int\left\vert \nabla\phi
^{-}\right\vert ^{2}dx\right)
\end{align*}
where $\Delta\phi^{-}=4\pi\int Df^{-}dv.$Notice that the last expression above
is the Antonov functional $4\pi H\left(  f^{-},f^{-}\right)  $. Since $f^{-}$
is spherical symmetric and odd in $v,$we have $H\left(  f^{-},f^{-}\right)
>0$ by the proof in \cite{KS85} which was further clarified in \cite{aly96}
and \cite{gr-king}. Therefore we get $k_{0}>0$ as desired and $\left(
A_{0}\phi,\phi\right)  \geq k_{0}\left\vert \phi\right\vert _{2}^{2}$.

To get the estimate (\ref{estimate-a0}), we rewrite
\begin{align*}
\left(  A_{0}\phi,\phi\right)   &  =\varepsilon\left(  \int\left\vert
\nabla\phi\right\vert ^{2}dx+4\pi\int\int f_{0}^{\prime}\left(  \phi
-\mathcal{P}\phi\right)  ^{2}dxdv\right)  +\left(  1-\varepsilon\right)
\left(  A_{0}\phi,\phi\right) \\
&  \geq\varepsilon\int\left\vert \nabla\phi\right\vert ^{2}dx-4\pi
\varepsilon\left\Vert \phi-\mathcal{P}\phi\right\Vert _{L_{\left\vert
f_{0}^{\prime}\right\vert }^{2}}^{2}+\left(  1-\varepsilon\right)
k_{0}\left\vert \phi\right\vert _{2}^{2}\\
&  \geq\varepsilon\int\left\vert \nabla\phi\right\vert ^{2}dx-8\pi
\varepsilon\left\Vert \phi\right\Vert _{L_{\left\vert f_{0}^{\prime
}\right\vert }^{2}}^{2}+\left(  1-\varepsilon\right)  k_{0}\left\vert
\phi\right\vert _{2}^{2}\text{ (since }\left\Vert \mathcal{P}\right\Vert
_{L_{\left\vert f_{0}^{\prime}\right\vert }^{2}\rightarrow L_{\left\vert
f_{0}^{\prime}\right\vert }^{2}}\leq1)\\
&  \geq\varepsilon\int\left\vert \nabla\phi\right\vert ^{2}dx+\left(  \left(
1-\varepsilon\right)  k_{0}-C\varepsilon\right)  \left\vert \phi\right\vert
_{2}^{2}\text{ }\geq\varepsilon\left(  \int\left\vert \nabla\phi\right\vert
^{2}dx+\left\vert \phi\right\vert _{2}^{2}\right)
\end{align*}
if $\varepsilon$ is small enough$.$
\end{proof}

Next, we will approximate the $\ker D$ by a finite dimensional approximation.
Let $\left\{  \xi_{i}(E,L)=\alpha_{i}(E)\beta_{i}(L)\right\}  _{i=1}^{\infty}$
be a smooth orthogonal basis for the subspace $\ker D=\left\{  g(E,L)\right\}
\subset$ $L_{\left\vert f_{0}^{\prime}\right\vert }^{2,r}.$Define the
finite-dimensional projection operator $\mathcal{P}_{N}:L_{\left\vert
f_{0}^{\prime}\right\vert }^{2,r}\rightarrow L_{\left\vert f_{0}^{\prime
}\right\vert }^{2,r}$ by%
\begin{equation}
\mathcal{P}_{N}h\equiv\sum_{i=1}^{N}(h,\xi_{i})_{\left\vert f_{0}^{\prime
}\right\vert }\xi_{i} \label{defn-Pn}%
\end{equation}
and the operator $A^{N}:H_{r}^{2}\rightarrow L_{r}^{2}$ by
\[
A^{N}\phi=-\Delta\phi+[4\pi\int f_{0}^{\prime}dv]\phi-4\pi\int f_{0}^{\prime
}\mathcal{P}_{N}\phi dv.
\]

\begin{lemma}
\label{lemma-an}There exists $K,\delta_{0}>0$ such that when $N>K$ we have
\begin{equation}
\left(  A^{N}\phi,\phi\right)  \geq\delta_{0}\left\vert \nabla\phi\right\vert
_{2}^{2} \label{estimate-An}%
\end{equation}
for any $\phi\in H_{r}^{2}$.
\end{lemma}

\begin{proof}
First we have $A^{N}\rightarrow A_{0}$ strongly in $L^{2}.$ In deed, for any
$\phi\in H_{r}^{2}$,
\[
\left\Vert A^{N}\phi-A_{0}\phi\right\Vert _{2}=\left\Vert \int4\pi
f_{0}^{\prime}\left(  \mathcal{P}_{N}\phi-\mathcal{P}\phi\right)
dv\right\Vert _{2}\leq C\left\Vert \mathcal{P}_{N}\phi-\mathcal{P}%
\phi\right\Vert _{L_{\left\vert f_{0}^{\prime}\right\vert }^{2}}\rightarrow0
\]
as $N\rightarrow\infty.$We claim that for $N$ sufficiently large, the lowest
eigenvalue of $A^{N}$ is at least $k_{0}/2$ where $k_{0}>0$ is the lowest
eigenvalue of $A_{0}$. Suppose otherwise, then there exists a sequence
$\left\{  \lambda_{n}\right\}  $ and $\left\{  \phi_{n}\right\}  \subset
H_{r}^{2}$ with $\lambda_{n}<k_{0}/2$, $\left\Vert \phi_{n}\right\Vert _{2}=1$
and $A^{n}\phi_{n}=\lambda_{n}\phi_{n}$. This implies that $\Delta\phi_{n}$ is
uniformly bounded in $L^{2}$, by elliptic estimate we have $\left\Vert
\phi_{n}\right\Vert _{H^{2}}\leq C$ for some constant $C$ independent of $n$.
Therefore there exists $\phi_{0}\in H_{r}^{2}$ such that $\phi_{n}%
\rightarrow\phi_{0}$ weakly in $H_{r}^{2}$. By the compact embedding of
$H_{r}^{2}$ $\hookrightarrow L_{r}^{2}$, we have $\phi_{n}\rightarrow\phi_{0}$
strongly in $L_{r}^{2}$ and $\left\Vert \phi_{0}\right\Vert _{2}=1$. The
strong convergence of $A^{n}\phi_{0}\rightarrow A_{0}\phi_{0}$ implies that
\[
A^{n}\phi_{n}\rightarrow A_{0}\phi_{0}%
\]
weakly in $L^{2}$. Let $\lambda_{n}\rightarrow\lambda_{0}\leq k_{0}/2$, then
we have $A_{0}\phi_{0}=\lambda_{0}\phi_{0}$, a contradiction. Therefore we
have $\left(  A^{N}\phi,\phi\right)  \geq k_{0}/2\left\vert \phi\right\vert
_{2}^{2}$ for $\phi\in H_{r}^{2},$ when $N$ is large enough. The estimate
(\ref{estimate-An}) is by the same proof of (\ref{estimate-a0}) in Lemma
\ref{lemma-ao}.
\end{proof}

Recalling (\ref{king}) with $f_{0}=[e^{E_{0}-E}-1]_{+\text{ }}$and
$Q_{0}(f)=(f+1)\ln(f+1)-f,$ we further define functionals (related to the
finite dimensional approximation of $\ker D$) as
\begin{align*}
A_{i}(f)  &  \equiv\int_{0}^{f}\alpha_{i}(-\ln(s+1)+E_{0})ds,\\
Q_{i}(f,L)  &  \equiv A_{i}(f)\beta_{i}(L),\text{ for }1\leq i\leq N.
\end{align*}
for $1\leq i\leq N.$ Clearly,
\[
\partial_{1}Q_{i}(f_{0},L)=\alpha_{i}(-\ln(f_{0}+1)+E_{0})\beta_{i}%
(L)=\alpha_{i}(E)\beta_{i}(L)=\xi_{i}(E,L),
\]
where $\left\{  \xi_{i}(E,L)\right\}  _{i=1}^{N}$ are used to define
$\mathcal{P}_{N}$ in Lemma \ref{lemma-an}. Define the Casimir functional
($E_{0}<0\,$)
\[
I(f)=\int[Q_{0}(f)+\frac{1}{2}|v|^{2}f-E_{0}f]dxdv-\frac{1}{8\pi}\int
|\nabla\phi|^{2}dx
\]
which is invariant of the nonlinear Vlasov-Poisson system. We introduce
additional $N$ invariants
\[
J_{i}(f,L)\equiv\int Q_{i}(f,L)dxdv.
\]
for $1\leq i\leq N$. We define $\Omega$ to be the support of $f_{0}(E).$ We
first consider
\begin{align*}
I(f)-I(f_{0})  &  =\int[Q_{0}(f)-Q_{0}(f_{0})+\frac{1}{2}|v|^{2}%
(f-f_{0})-E_{0}(f-f_{0})]dxdv\\
&  \ \ \ \ \ -\frac{1}{4\pi}\int\nabla U_{0}\cdot\nabla(U-U_{0})-\frac{1}%
{8\pi}\int|\nabla(U-U_{0})|^{2}dx\\
&  =\int[Q_{0}(f)-Q_{0}(f_{0})+(E-E_{0})(f-f_{0})]dxdv-\frac{1}{8\pi}%
\int|\nabla(U-U_{0})|^{2}dx.
\end{align*}
We define
\[
g=f-f_{0},\text{ \ \ \ \ \ }\phi=U-U_{0}%
\]
and
\[
g_{\text{in}}\equiv(f-f_{0})\mathbf{1}_{\Omega},\text{ \ \ \ \ \ \ }%
g_{\text{out}}\equiv(f-f_{0})\mathbf{1}_{\Omega^{c}},\text{ \ \ }\Delta
\phi_{\text{in}}\equiv\int g_{\text{in}},\text{ \ \ \ }\Delta\phi_{\text{out}%
}\equiv\int g_{\text{out }}.
\]
And we define the distance function for nonlinear stability as
\begin{align}
d(f,f_{0})  &  \equiv\left\{  \int\int[Q_{0}(g_{\text{in}}+f_{0})-Q_{0}%
(f_{0})+(E-E_{0})g_{\text{in}}]dxdv\right\}  +\frac{1}{8\pi}\int|\nabla
\phi_{\text{in}}|^{2}dx\label{defn-distance}\\
&  +\left\{  \int\int Q_{0}(g_{\text{out}})dxdv+\int_{E\geq E_{0}}%
(E-E_{0})g_{\text{out}}dxdv\right\} \nonumber\\
&  =d_{\text{in}}+\frac{1}{8\pi}\int|\nabla\phi_{\text{in}}|^{2}%
dx+d_{\text{out}},\nonumber
\end{align}
for which each term is non-negative. We therefore split:%
\begin{align*}
&  I(f)-I(f_{0})\\
&  =\left\{  \int[Q_{0}(f_{0}+g_{\text{in}})-Q_{0}(f_{0})+(E-E_{0}%
)g_{\text{in}}]dxdv-\frac{1}{8\pi}\int|\nabla\phi_{\text{in}}|^{2}dx\right\}
+\\
&  \left\{  \int Q_{0}(g_{\text{out}})dxdv+\int_{E\geq E_{0}}(E-E_{0}%
)g_{\text{out}}dxdv-\frac{1}{8\pi}\int|\nabla\phi_{\text{out}}|^{2}dx-\frac
{1}{4\pi}\int\nabla\phi_{\text{out}}\cdot\nabla\phi_{\text{in}}dx\right\} \\
&  =I_{\text{in}}+I_{\text{out }}.
\end{align*}
In the estimates below, we use $C,C^{\prime},C^{\prime\prime}$ to denote
general constants depending only on $f_{0}\ $and quantities like $\left\Vert
f\left(  t\right)  \right\Vert _{L^{p}}$ $\left(  p\in\left[  1,+\infty
\right]  \right)  $ which equals $\left\Vert f\left(  0\right)  \right\Vert
_{L^{p}}$ and therefore always under control. We first estimate $\left\Vert
\nabla\phi_{\text{out}}\right\Vert _{2}^{2}$ to be of higher order of $d$,
which also implies that $\int\nabla\phi_{\text{out}}\cdot\nabla\phi
_{\text{in}}dx$ is of higher order of $d$.

\begin{lemma}
\label{lemma-out}For $\varepsilon>0\ $sufficiently small, we have
\[
\int|\nabla\phi_{\text{out}}|^{2}dx\leq C\left(  \varepsilon d(f,f_{0}%
)+\frac{1}{\varepsilon^{5/3}}[d(f,f_{0})]^{5/3}\right)  .
\]

\end{lemma}

\begin{proof}
In fact, since
\begin{align*}
\int|\nabla\phi_{\text{out}}|^{2}dx  &  \leq C||\int g_{\text{out }%
}dv||_{L^{6/5}}^{2}\\
&  \leq C||\int g_{\text{out }}\mathbf{1}_{E_{0}\leq E\leq E_{0}+\varepsilon
}dv||_{L^{6/5}}^{2}+C||\int g_{\text{out }}\mathbf{1}_{E>E_{0}+\varepsilon
}dv||_{L^{6/5}}^{2}.
\end{align*}
The first term is bounded by%
\begin{align*}
&  \left[  \int[\int g_{\text{out }}^{2}dv]^{3/5}[\int\mathbf{1}_{E_{0}\leq
E\leq E_{0}+\varepsilon}dv]^{3/5}dx\right]  ^{5/3}\\
&  \leq\lbrack\int g_{\text{out }}^{2}dvdx]\times\left[  \int[\int
\mathbf{1}_{E_{0}\leq E\leq E_{0}+\varepsilon}dv]^{3/2}dx\right]  ^{2/3}\\
&  \leq C\varepsilon\lbrack\int g_{\text{out }}^{2}dvdx]\leq C\varepsilon
\lbrack\int g_{\text{out }}^{2}dvdx]\\
&  \leq C\varepsilon d(f,f_{0}).
\end{align*}
In the above estimates, we use that $\int\int Q_{0}(g_{\text{out}})dvdx\geq
c\int g_{\text{out }}^{2}dvdx$ and
\[
\int\mathbf{1}_{E_{0}\leq E\leq E_{0}+\varepsilon}dv\leq C\varepsilon,
\]
which can be checked by an explicit computation when $\varepsilon>0\ $is
sufficiently small such that $E_{0}+\varepsilon\leq0$.

On the other hand, by the standard estimates (see \cite[P. 120-121]{glassey})
\begin{align*}
&  ||\int g_{\text{out }}\mathbf{1}_{E>E_{0}+\varepsilon}dv||_{L^{6/5}}^{2}\\
&  \leq\left[  \int\int g_{\text{out \ }}\mathbf{1}_{E>E_{0}+\varepsilon
}dxdv\right]  ^{\frac{7}{6}}\times\left[  \int\int|v|^{2}g_{\text{out \ }%
}\mathbf{1}_{E>E_{0}+\varepsilon}dxdv\right]  ^{\frac{1}{2}}\\
&  \leq\left[  \frac{1}{\varepsilon}\int\int(E-E_{0})g_{\text{out \ }%
}\mathbf{1}_{E>E_{0}+\varepsilon}dxdv\right]  ^{\frac{7}{6}}\\
&  \times\left[  \int\int(E-E_{0})g_{\text{out \ }}\mathbf{1}_{E>E_{0}%
+\varepsilon}dxdv+2\sup|U_{0}|\int\int g_{\text{out \ }}\mathbf{1}%
_{E>E_{0}+\varepsilon}dxdv\right]  ^{\frac{1}{2}}\\
&  \leq\left(  \frac{1}{\varepsilon}d\right)  ^{\frac{7}{6}}\left(
d+\frac{2\sup|U_{0}|}{\varepsilon}d\right)  ^{\frac{1}{2}}\leq\frac
{C}{\varepsilon^{5/3}}d^{5/3}\text{.}%
\end{align*}

\end{proof}

By Lemma \ref{lemma-out}, we have%
\begin{align*}
\left\vert \int\nabla\phi_{\text{out}}\cdot\nabla\phi_{\text{in}%
}dx\right\vert  &  \leq\left\Vert \nabla\phi_{\text{out}}\right\Vert
_{2}\left\Vert \nabla\phi_{\text{in}}\right\Vert _{2}\\
&  \leq C\left(  \varepsilon^{1/3}d(f,f_{0})+\frac{1}{\varepsilon^{5/6}%
}[d(f,f_{0})]^{4/3}\right)
\end{align*}
and therefore for $\varepsilon$ sufficiently small,
\begin{equation}
I_{\text{out }}\geq d_{\text{out}}-C\left(  \varepsilon^{1/3}d(f,f_{0}%
)+\frac{1}{\varepsilon^{5/6}}[d(f,f_{0})]^{4/3}+\frac{1}{\varepsilon^{5/3}%
}[d(f,f_{0})]^{5/3}\right)  . \label{estimate-I-out}%
\end{equation}

To estimate $I_{\text{in}}$, we split it into three parts:
\begin{align}
&  \tau\left\{  \int[Q_{0}(f_{0}+g_{\text{in}})-Q_{0}(f_{0})+(E-E_{0}%
)g_{\text{in}}+\phi_{\text{in}}g_{\text{in}}]dxdv+\frac{1}{8\pi}\int
|\nabla\phi_{\text{in}}|^{2}dx\right\}  +\nonumber\\
&  (1-\tau)\left\{  \int[Q_{0}(f_{0}+g_{\text{in}})-Q_{0}(f_{0})+(E-E_{0}%
)g_{\text{in}}+(I-P_{N})\phi_{\text{in}}g_{\text{in}}]dxdv+\frac{1}{8\pi}%
\int|\nabla\phi_{\text{in}}|^{2}dx\right\} \nonumber\\
&  +(1-\tau)\int P_{N}\phi_{\text{in}}g_{\text{in}}dxdv\nonumber\\
&  =I_{\text{in}}^{1}+I_{\text{in}}^{2}+I_{\text{in}}^{3}, \label{i1+i2+i3}%
\end{align}
where $\Delta\phi_{\text{in}}=4\pi\int g_{\text{in }}dv.$ We estimate each
term in the following lemmas.

\begin{lemma}%
\begin{equation}
I_{\text{in}}^{1}\geq\frac{\tau}{2}d_{\text{in}}-C\tau\int|\nabla
\phi_{\text{in}}|^{2}dx. \label{estimate-I-in1}%
\end{equation}

\end{lemma}

\begin{proof}
In fact, since the integration region $\Omega$ is finite, we have
\begin{align*}
I_{\text{in}}^{1}=  &  \tau\left[  \int\int[Q_{0}(f_{0}+g_{\text{in}}%
)-Q_{0}(f_{0})+(E-E_{0})g_{\text{in}}+\phi_{\text{in}}g_{\text{in}}%
]dxdv+\frac{1}{8\pi}\int|\nabla\phi_{\text{in}}|^{2}dx\right] \\
&  \geq\tau\int\int[Q_{0}(f_{0}+g_{\text{in}})-Q_{0}(f_{0})+(E-E_{0}%
)g_{\text{in}}]dxdv-C\tau||\phi_{\text{in}}||_{L^{6}}||g_{\text{in}%
}||_{L^{6/5}}\\
&  \geq\tau\int\int[Q_{0}(f_{0}+g_{\text{in}})-Q_{0}(f_{0})+(E-E_{0}%
)g_{\text{in}}]dxdv-C^{\prime}\tau||\nabla\phi_{\text{in}}||_{L^{2}%
}||g_{\text{in}}||_{2}\\
&  \geq\frac{\tau}{2}d_{\text{in}}-C^{\prime\prime}\tau||\nabla\phi
_{\text{in}}||_{2}^{2},
\end{align*}
since%
\[
d_{\text{in}}=\int[Q_{0}(f_{0}+g_{\text{in}})-Q_{0}(f_{0})+(E-E_{0}%
)g_{\text{in}}]dxdv\geq C||g_{\text{in}}||_{2}^{2}.
\]

\end{proof}

To estimate $I_{\text{in}}^{2}$, we need the following pointwise duality lemma
from elementary calculus.

\begin{lemma}
\label{lemma-duality}For any $c,$ and any $h,$ we have
\[
g_{c,f_{0}}\left(  h\right)  =Q_{0}(h+f_{0})-Q_{0}(f_{0})-Q_{0}^{\prime}%
(f_{0})h-ch\geq(f_{0}+1)(1+c-e^{c}).
\]

\end{lemma}

\begin{proof}
Direct computation yields that the minimizer $f_{c}$ of $g_{c,f_{0}}\left(
h\right)  $ satisfies the Euler-Lagrange equation%
\[
\ln\left(  f_{c}+f_{0}+1\right)  -\ln\left(  f_{0}+1\right)  -c=0,
\]
so%
\[
f_{c}=\left(  f_{0}+1\right)  \left(  e^{c}-1\right)  .
\]
Thus by using the Euler-Lagrange equation, we deduce
\begin{align*}
\min g_{c,f_{0}}\left(  h\right)   &  =g_{c,d}\left(  f_{c}\right) \\
&  =(f_{c}+f_{0}+1)\ln(1+f_{c}+f_{0})\\
&  -(f_{0}+1)\ln(1+f_{0})-[1+\ln(f_{0}+1)]f_{c}-cf_{c}\\
&  =(f_{c}+f_{0}+1)[\ln(1+f_{c}+f_{0})-\ln(f_{0}+1)-c]\\
&  +f_{c}\ln(1+f_{0})+c(f_{0}+1)-[1+\ln(f_{0}+1)]f_{c}\\
&  =(f_{0}+1)(1+c-e^{c}).
\end{align*}

\end{proof}

\begin{lemma}%
\begin{equation}
I_{\text{in}}^{2}\geq\frac{\left(  1-\tau\right)  \delta_{0}}{8\pi}\int
|\nabla\phi_{\text{in}}|^{2}dx-Ce^{C^{\prime}d^{\frac{1}{2}}}d^{\frac{3}{2}}.
\label{estimate-I-in2}%
\end{equation}

\end{lemma}

\begin{proof}
Recall (\ref{i1+i2+i3}). By using Lemma \ref{lemma-duality} for $c=-\left(
\phi_{\text{in}}-P_{N}\phi_{\text{in}}\right)  $ and using the Taylor
expansion, we have
\begin{align*}
I_{\text{in}}^{2}  &  =(1-\tau)\int\int[Q_{0}(f_{0}+g_{\text{in}})-Q_{0}%
(f_{0})+(E-E_{0})g_{\text{in}}+\left(  \phi_{\text{in}}-P_{N}\phi_{\text{in}%
}\right)  f_{\text{in}}]dxdv\\
&  \text{ }\ \ \ \ \ +\frac{1}{8\pi}(1-\tau)\int|\nabla\phi_{\text{in}}%
|^{2}dx\\
&  \geq\frac{1}{8\pi}(1-\tau)\int|\nabla\phi_{\text{in}}|^{2}dx+(1-\tau
)\int\int(f_{0}+1)\mathbf{1}_{\Omega}(1+\phi_{\text{in}}-P_{N}\phi_{\text{in}%
}-e^{\phi_{\text{in}}-P_{N}\phi_{\text{in}}})dxdv\\
&  \geq\frac{1-\tau}{8\pi}\left\{  \int|\nabla\phi_{\text{in}}|^{2}dx-4\pi
\int\int\left\vert f_{0}^{\prime}\left(  E\right)  \right\vert \left(
\phi_{\text{in}}-P_{N}\phi_{\text{in}}\right)  ^{2}dxdv\right\}  \text{ }\\
&  \ \ \ \ -Ce^{\left\vert \phi_{\text{in}}-P_{N}\phi_{\text{in}}\right\vert
_{\infty}}\int\int\left\vert f_{0}^{\prime}\left(  E\right)  \right\vert
\left\vert \phi_{\text{in}}-P_{N}\phi_{\text{in}}\right\vert ^{3}dxdv\text{
}\ \ \ \text{(Note }\left(  f_{0}(E)+1\right)  \mathbf{1}_{\Omega}%
=|f_{0}^{\prime}(E)|\text{)}\\
&  \geq\frac{\left(  1-\tau\right)  \delta_{0}}{8\pi}\int|\nabla
\phi_{\text{in}}|^{2}dx-Ce^{\left\vert \phi_{\text{in}}-P_{N}\phi_{\text{in}%
}\right\vert _{\infty}}\int\int\left\vert f_{0}^{\prime}\left(  E\right)
\right\vert \left\vert \phi_{\text{in}}-P_{N}\phi_{\text{in}}\right\vert
^{3}dxdv.
\end{align*}
In the last line, we have used Lemma \ref{lemma-an}. To estimate the last term
above and conclude our lemma, it suffices to show
\[
|\phi_{\text{in}}-P_{N}\phi_{\text{in}}|_{\infty}\leq C_{N}d^{\frac{1}{2}}.
\]
This follows from the facts that for the fixed $N$ smooth functions $\xi_{i},$
we have
\[
\left\vert P_{N}\phi_{\text{in}}\right\vert _{\infty}=\left\vert \sum
_{i=1}^{N}(\phi_{\text{in}},\xi_{i})_{\left\vert f_{0}^{\prime}\right\vert
}\xi_{i}\right\vert _{\infty}\leq C_{N}\left\vert \phi_{\text{in}}\right\vert
_{\infty},
\]
and since $\phi$ is spherically symmetric,
\begin{align*}
|\phi_{\text{in}}|\left(  r\right)   &  =\left\vert \frac{1}{r}\int_{0}%
^{r}u^{2}\rho_{\text{in}}\left(  u\right)  du+\int_{r}^{R}u\rho_{\text{in}%
}\left(  u\right)  du\right\vert \\
&  \leq C^{\prime}\sqrt{R}\left\vert \rho_{\text{in}}\right\vert _{2}\leq
C^{\prime\prime}\left\Vert g_{\text{in}}\right\Vert _{2}\leq C_{N}d^{\frac
{1}{2}}%
\end{align*}
where $\rho_{\text{in}}=\int g_{\text{in}}dv$ and $R$ is the support radius of
$\rho_{\text{in}}$.
\end{proof}

We now estimate the term $\int\int P_{N}\phi_{\text{in}}f_{\text{in}}dxdv$,
for which we use the additional invariants.

\begin{lemma}
For any $\varepsilon>0,$ we have
\begin{equation}
\left\vert I_{\text{in}}^{3}\right\vert \leq C(d^{1/2}(0)+\varepsilon
^{1/2}d^{1/2}+\frac{1}{\varepsilon}d)d^{1/2}. \label{estimate-I-in3}%
\end{equation}

\end{lemma}

\begin{proof}
By the definition of $I_{\text{in}}^{3}$ in (\ref{i1+i2+i3})$,$ it suffices to
estimate $(g_{\text{in}},\xi_{i}).$ We expand
\begin{align*}
&  J_{i}(f,L)-J_{i}(f_{0},L)\\
&  =J_{i}(f_{0}+g_{\text{in}},L)-J_{i}(f_{0},L)+J_{i}(g_{\text{out}},L)\\
&  =(g_{\text{in }},\xi_{i})+O(d)+J_{i}(g_{\text{out}},L).
\end{align*}
Notice that
\begin{align*}
|J_{i}(g_{\text{out}},L)|  &  \leq C||g_{\text{out}}||_{L^{1}}\leq
C||\mathbf{1}_{\left\{  E_{0}\leq E\leq E_{0}+\varepsilon\right\}
}g_{\text{out}}||_{L^{1}}+C||\mathbf{1}_{\left\{  E\geq E_{0}+\varepsilon
\right\}  }g_{\text{out}}||_{L^{1}}\\
&  \leq\varepsilon^{1/2}||g_{\text{out}}||_{L^{2}}+\frac{C}{\varepsilon
}||\mathbf{1}_{\left\{  E\geq E_{0}+\varepsilon\right\}  }(E-E_{0}%
)g_{\text{out}}||_{L^{1}}\leq C[\varepsilon^{1/2}d^{1/2}+\frac{1}{\varepsilon
}d].
\end{align*}
It thus follows that
\begin{align*}
|(g_{\text{in }},\xi_{i})|  &  \leq|J_{i}(f(0),L)-J_{i}(f_{0}%
,L)|+C[\varepsilon^{1/2}d^{1/2}+\frac{1}{\varepsilon}d]\\
&  \leq C[d^{1/2}(0)+\varepsilon^{1/2}d^{1/2}+\frac{1}{\varepsilon}d].
\end{align*}

Therefore%
\begin{align*}
\left\vert I_{\text{in}}^{3}\right\vert  &  =\left(  1-\tau\right)  \left\vert
\int\int P_{N}\phi_{\text{in}}g_{\text{in }}dxdv\right\vert =\left\vert
\int\int\left(  \sum_{i=1}^{N}(\phi_{\text{in}},\xi_{i})_{\left\vert
f_{0}^{\prime}\right\vert }\xi_{i}\right)  g_{\text{in }}dxdv\right\vert \\
&  \leq\sum_{i=1}^{N}\left\vert (\phi_{\text{in}},\xi_{i})_{\left\vert
f_{0}^{\prime}\right\vert }\right\vert |(\xi_{i},g_{\text{in}})|\leq
C^{\prime}\sum_{i=1}^{N}\left\vert \phi_{\text{in}}\right\vert _{\infty}%
|(\xi_{i},g_{\text{in}})|\\
&  \leq Cd^{1/2}[d^{1/2}(0)+\varepsilon^{1/2}d^{1/2}+\frac{1}{\varepsilon}d].
\end{align*}

\end{proof}

Now we prove the nonlinear stability of King model.

\begin{proof}
[Proof of Theorem \ref{theorem-stability}]The global existence of classical
solutions of 3D Vlasov-Poisson system was shown in \cite{phaffa} for compactly
supported initial data $f\left(  0\right)  \in C_{c}^{1}$. Let the unique
global solution be $\left(  f\left(  t\right)  ,\phi\left(  t\right)  \right)
$. Let $d\left(  t\right)  =d(f\left(  t\right)  ,f_{0})$. Combining estimates
(\ref{estimate-I-out}), (\ref{estimate-I-in1}), (\ref{estimate-I-in2}) and
(\ref{estimate-I-in3}), we have
\begin{align*}
I(f\left(  0\right)  )-I(f_{0})  &  =I(f\left(  t\right)  )-I(f_{0})\\
&  \geq d_{\text{out}}+\frac{\tau}{2}d_{\text{in}}+\left(  \frac{\left(
1-\tau\right)  \delta_{0}}{8\pi}-C\tau\right)  \int|\nabla\phi_{\text{in}%
}|^{2}dx\\
&  -C\left(  \varepsilon^{1/3}d\left(  t\right)  +\frac{1}{\varepsilon^{5/6}%
}d\left(  t\right)  ^{4/3}+\frac{1}{\varepsilon^{5/3}}d\left(  t\right)
^{5/3}\right)  -Ce^{C^{\prime}d\left(  t\right)  ^{\frac{1}{2}}}d\left(
t\right)  ^{\frac{3}{2}}\\
&  -Cd\left(  t\right)  ^{1/2}[d^{1/2}(0)+\varepsilon^{1/2}d\left(  t\right)
^{1/2}+\frac{1}{\varepsilon}d\left(  t\right)  ].
\end{align*}
Thus by choosing $\varepsilon$ and $\tau$ sufficiently small, there exists
$\delta^{\prime}>0$ such that
\begin{align}
I(f\left(  0\right)  )-I(f_{0})  &  \geq\delta^{\prime}d(t)-C\left(  d\left(
t\right)  ^{4/3}+d\left(  t\right)  ^{5/3}+d\left(  t\right)  ^{3/2}\right)
-Ce^{C^{\prime}d\left(  t\right)  ^{\frac{1}{2}}}d\left(  t\right)  ^{\frac
{3}{2}}\label{estimate-ec}\\
&  -Cd\left(  t\right)  ^{1/2}d^{1/2}(0).\nonumber
\end{align}
It is easy to show that $I(f\left(  0\right)  )-I(f_{0})\leq C^{\prime\prime
}d\left(  0\right)  $. Define the functions $y_{1}\left(  x\right)
=\delta^{\prime}x^{2}-Ce^{C^{\prime}x}x^{3}-C\left(  x^{8/3}+x^{10/3}%
+x^{3}\right)  $ and $y_{2}\left(  x\right)  =Cd\left(  0\right)
^{1/2}x+C^{\prime\prime}d\left(  0\right)  $. Then above estimates implies
that $y_{1}\left(  d\left(  t\right)  ^{1/2}\right)  \leq y_{2}\left(
d\left(  t\right)  ^{1/2}\right)  $. The function $y_{1}$ is increasing in
$\left(  0,x_{0}\right)  $ where $x_{0}$ is the first maximum point. So if
$d\left(  0\right)  $ is sufficiently small, the line $y=y_{2}\left(
x\right)  $ intersects the curve $y=y_{1}\left(  x\right)  $ at points
$x_{1},x_{2},\cdots,$ with $x_{1}\left(  d\left(  0\right)  \right)
<x_{0}<x_{2}\left(  d\left(  0\right)  \right)  <\cdots$. Thus the inequality
$y_{1}\left(  x\right)  \leq y_{2}\left(  x\right)  $ is valid in disjoint
intervals $\left[  0,x_{1}\left(  d\left(  0\right)  \right)  \right]  $ and
$[x_{2}\left(  d\left(  0\right)  \right)  ,x_{3}\left(  d\left(  0\right)
\right)  ],\cdots.$ Because $d\left(  t\right)  $ is continuous, we have that
$d\left(  t\right)  ^{1/2}<x_{1}\left(  d\left(  0\right)  \right)  $ for all
$t<\infty$, provided we choose $d\left(  0\right)  ^{1/2}<x_{0}$. Since
$x_{1}\left(  d\left(  0\right)  \right)  \rightarrow0$ as $d\left(  0\right)
\rightarrow0$, we deduce the nonlinear stability in terms of the distance
functional $d\left(  t\right)  ^{1/2}$.
\end{proof}

\vskip 0.5cm

\noindent\textbf{Acknowledgements}{\Large \ }

\vskip 0.2cm This research is supported partly by NSF grants DMS-0603815 and
DMS-0505460. We thank the referees for comments and corrections.

\end{document}